\documentclass[10pt,twocolumn,twoside]{IEEEtran}
\usepackage{calc}

\usepackage{multirow}
\usepackage{cite}
\usepackage{graphicx,subfigure}
\usepackage{psfrag}
\usepackage{amsmath,amssymb}
\usepackage{color}

\DeclareMathOperator*{\defeq}{\triangleq}

\newtheorem{theorem}{Theorem}
\newtheorem{corollary}{Corollary}[theorem]
\newcommand{\bit}{\begin{itemize}}
\newcommand{\eit}{\end{itemize}}

\newcommand{\bc}{\begin{center}}
\newcommand{\ec}{\end{center}}

\newcommand{\ba}{\begin{array}}
\newcommand{\ea}{\end{array}}

\newcommand{\beq}{\begin{equation}}
\newcommand{\eeq}{\end{equation}}

\newcommand{\beqn}{\begin{equation*}}
\newcommand{\eeqn}{\end{equation*}}

\newcommand{\bean}{\begin{eqnarray*}}
\newcommand{\eean}{\end{eqnarray*}}
\newcommand{\bea}{\begin{eqnarray}}
\newcommand{\eea}{\end{eqnarray}}

\def\C{\mathbb{C}}

\def\F{\mathbb{F}}

\newcommand{\Xc}{{\mathcal X}}

\newcommand{\T}{{\scriptscriptstyle\mathsf{T}}}

\newtheorem{remark}{Remark}

\renewcommand{\Bmatrix}[1]{\begin{bmatrix}#1\end{bmatrix}}

\newtheorem{claim}{Claim}

\begin{document}
\sloppy

\title{Feedback through Overhearing}
\author{Jinyuan Chen,  Ayfer \"Ozg\"ur and Suhas Diggavi 
\thanks{Jinyuan Chen and Ayfer \"Ozg\"ur  are with Stanford University, CA (emails:jinyuanc@stanford.edu, aozgur@stanford.edu ). Suhas Diggavi is with University of California, Los Angeles (email:suhas@ee.ucla.edu). }
}

%%%%%%%%%%%%%%%%%%%%%%%%%%%%%%%%%%%%%%%%%%%%%

\maketitle
\thispagestyle{empty}
%\pagestyle{headings}

%%%%%%%%%%%%%%%%%%%%%%%%%%%%%%%%%%%%%%%%%%%%% 
\begin{abstract}
  In this paper we examine the value of feedback that comes from
  overhearing, without dedicated feedback resources. We focus on a
  simple model for this purpose: a deterministic two-hop interference
  channel, where feedback comes from overhearing the forward-links.  A
  new aspect brought by this setup is the dual-role of the relay
  signal. While the relay signal needs to convey the source message to
  its corresponding destination, it can also provide a feedback signal
  which can potentially increase the capacity of the first hop. We derive inner  and outer bounds on the sum capacity which match for a large range of the parameter values. Our results identify the parameter ranges where overhearing  can provide non-negative capacity gain and can even achieve the performance with dedicated-feedback resources. The results also provide insights into which transmissions are most useful to overhear.

\end{abstract}

\section{Introduction}

Shannon showed that feedback cannot increase the capacity of the
point-to-point discrete memoryless channel \cite{shannon:56}.  Later
works showed that feedback can increase the capacity of the Gaussian
MAC, broadcast and relay channels \cite{GaarderWolf:75, CG:79,
  Ozarow:84}, but only through a power gain. More recently, it has
been shown in \cite{ST:11} that feedback can provide
degrees-of-freedom gain in interference channels, which has generated
significant recent interest in this setup
\cite{VSA:12,WSDV:13,PTPH:13,SAYS:09,SWT:12,CD:12}.

A common assumption in all these works is that there is a dedicated
channel for feedback, which in practice corresponds to allocating some
of the time/frequency resources of the system for feeding back
information. However, multi-hop wireless communication brings the possibility
of feedback, through overhearing, which does not require any dedicated
resources; transmitters of the previous hop can overhear the
transmissions in the next hop, which can be used as feedback. Our goal
in this paper is to understand how to optimally exploit such
overhearing and the corresponding benefits it can provide.

To study this question,  we focus on a two-hop
interference channel where two sources are communicating to their
corresponding destinations by the help of two relay nodes. See
Fig.~\ref{fig:TwoHopIC}. We assume that the transmissions of the
relays in the second hop can be overheard by the sources in the first
hop. To simplify the analysis, we assume that the relay transmissions
over the second hop do not interfere with each other and we
approximate the Gaussian channels in the system by a linear
deterministic model \cite{ADT:11}. We characterize the sum capacity of
this setup for most parameter ranges and provide lower and upper
bounds on the capacity for the remaining regimes.

This setup reveals the dual role of the relay signal: its need to
convey the message to its corresponding destination versus its
potential to send a feedback signal and increase the capacity of the
first hop.  Since it is well-understood that feedback increases the
capacity of the interference channel by allowing each transmitter to
(partially) learn the message of the other transmitter, this
introduces a trade-off between these two roles of the relay signal. On
one hand, the relay signal should contain information about the
desired message at its corresponding destination, on the other hand it has to
contain information about the interfering message in order to increase
the capacity of the first hop.

\begin{figure}[t!]
\centering
\includegraphics[width=8.5cm]{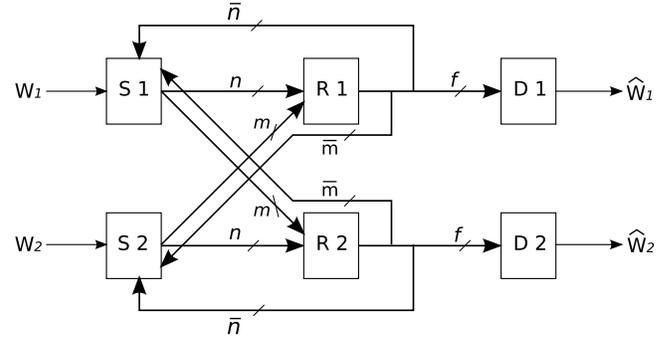}
\caption{Deterministic model of a two-hop interference channel. The
  labels indicate the strengths of the corresponding links. The links
  arriving to the same node are subject to superposition and the links
  emanating from the same node are subject to broadcast (i.e. they
  carry the same signal possibly at different strengths.)}
\label{fig:TwoHopIC}
\end{figure}

The main contribution of our paper is to design strategies that are able to optimally consolidate these two roles of the relay signal in various regimes. These strategies make use of the ideas for the interference channel with feedback from \cite{ST:11,VSA:12, SWT:12}, however are strictly different since the main focus in our case is consolidating feedforward and feedback communication. Our upper bounds also focus on this new aspect trying to capture the best trade-off one can have between these two roles in various regimes. 

Our results show that depending on the regime, these two roles of the
relay signal can be either compatible or competing. For example, when
the interference channel in the first hop is in the weak interference
regime and there is a strong backward cross-link from the relays to
the source nodes (i.e., from R1 to S2 and R2 to S1), by transmitting
the bits of its corresponding source node, R1 can simultaneously
communicate these bits to their final destination D1 and feed them back
to the other source node S2 which gives S2 the opportunity to (partially) learn
the message of S1 (the same holds for R2, D2 and S2 respectively). 
In the next transmission, S2 can
use these bits to resolve part of the interference at R2 which
helps to increase the capacity of the first hop, without sacrificing
the capacity of the second hop for feedback. On the other hand, when
the interference channel in the first hop is in the strong
interference regime and there is a strong backward direct link
from the relays to the source nodes (i.e., from R1 to S1 and R2 to
S2), the relay signal needs to convey the desired message to its
corresponding destination while conveying the other source message to
its own source node (for example R1 needs to convey the message of S1
to D1 and the message of S2 to S1). 
%We show that the best that can be
%done in this case is rate splitting between these two roles of the
%relay signal. 
We show that in this regime rate splitting between these two relaying roles is
optimal.  As a consequence, we identify the regimes where overheard
feedback can do as well as dedicated feedback channels (of the same
capacity) and also the regimes where the capacity suffers from the
overheard nature of the feedback.

\begin{figure}[t!]
\centering
\includegraphics[width=7cm]{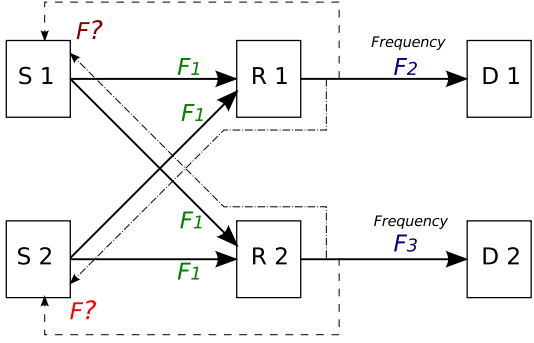}
\caption{System model with three frequencies: The first hop operates
  over frequency $F_1$, and the relays operate over frequencies $F_2$
  and $F_3$ respectively.  The source nodes have the option to listen
  to the second hop transmissions by tuning to either frequency $F_2$
  or $F_3$.}
\label{fig:F1F2F3TwoHopIC}
\end{figure}

Our results also suggest insights on which transmissions are more
useful to overhear. Consider the setup in
Fig.~\ref{fig:F1F2F3TwoHopIC}, where the first hop operates over
frequency $F_1$ and the relays operate over frequencies $F_2$ and
$F_3$ respectively. Assume that the source nodes have the option to
listen to the second hop transmissions by tuning to either frequency
$F_2$ or $F_3$. Assuming all backward channels have the same strength,
when the interference channel in the first hop is in the weak
interference regime our results reveal that it is more advantageous in terms of sum capacity for S1
to listen to R2's transmission over $F_3$ and for S1 to listen to R1's
transmission over $F_2$ rather than vice versa. (In other words, the
sum capacity of having $( \bar{m} = \theta, \bar{n} =0, m, n, f)$ is
better than that of having $( \bar{m} = 0, \bar{n} =\theta, m, n, f)$
for any $\theta, m,n,f$ s.t. $m/n\leq 2/3$ where $m,n,\bar{m},\bar{n},f$ indicate the strength of the deterministic channels as indicated in Figure~\ref{fig:F1F2F3TwoHopIC}.) This is fundamentally due
to the overheard nature of the feedback signal; feeding back over the
cross-link is compatible with feedforward communication, while feeding
back over the direct-link leads to rate splitting and therefore
smaller overall sum rate.  On the other hand, in the strong interference regime, it is
more desirable for S1 to listen to R1 over $F_2$ and for S2 to listen
to R2's transmission over $F_3$.

%\subsection{Structure}

The remainder of this work is organized as follows. 
Section~\ref{sec:system} first describes the system model. Section~\ref{sec:result} then gives the main results of this work. 
%The  sketches of achievability proof are shown in  Section~\ref{sec:scheme}, leaving some details in the appendix and the rest details in \cite{COD:14} due to the lack of space, while  the outer bound proof is  shown in Section~\ref{sec:outerbound}. 
The  sketches of the achievability proof are given in  Section~\ref{sec:scheme}, leaving the details to the appendix, while  the proof of the outer bound is given in Section~\ref{sec:outerbound}. Finally, Section~\ref{sec:conclusion} comes to the conclusions.

\section{System model  \label{sec:system}}

We consider a two-hop deterministic interference channel where the
source nodes S1 and S2 want to communicate to destination nodes D1 and
D2 respectively, with the help of two relay nodes R1 and R2.  Let $X_{S1}, X_{S2},
X_{R1}, X_{R2}$ denote the signals transmitted by S1, S2, R1 and R2
respectively and let $Y_{S1}, Y_{S2}, Y_{R1}, Y_{R2}, Y_{D1}, Y_{D2}$
similarly denote the signals received by the corresponding nodes.  The
input and output relations between these signals at time $t$ are given
as follows:
\begin{align}  
Y_{R1}[t] &=   S^{q- n} X_{S1}[t]  \oplus   S^{q- m} X_{S2}[t]          \label{eq:detTHIC1} \\
Y_{R2}[t] &=   S^{q- m} X_{S1}[t]  \oplus   S^{q- n} X_{S2}[t]           \label{eq:detTHIC2} 
\end{align}  
where $X_{S1}[t] = \Bmatrix{ X_{S1,1}[t], \cdots, X_{S1,q}[t] }^\T \in
\F^{q}_{2}$ , $q=\max(m,n)$, and $S^{q- n}$ is a $q\times q$ lower
shift matrix of the form
$$
S^{q- n}=\left[\begin{array}{lllllll} 0 &0 &0 &\cdots &\cdots &\cdots &0\\
\vdots & \vdots &\vdots &\vdots &\vdots &\vdots&\vdots\\
1 &0 &0 &\cdots &\cdots &\cdots &0\\
0 &1 &0 &\cdots&\cdots &\cdots  &0\\
\vdots & \vdots &\vdots &\vdots &\vdots &\vdots &\vdots\\
0 &0 &0 &\cdots &1& \cdots &0\\
\end{array}\right]
$$
which upon multiplying $X_{S1}[t]$ yields $\Bmatrix{ 0,
  \cdots,0, X_{S1,1}[t], \cdots, X_{S1,n}[t] }^\T$,
($X_{S2}[t],Y_{R1}[t],Y_{R2}[t]$ and $S^{q- m}$ are similarly
defined). $\oplus$ denotes modulo~2 addition. (See \cite{ADT:11} for
more detailed description of the model). Note that $n$ and $m$ denote
the number of bits that can be communicated over the direct and
cross links respectively over the first hop (see
Fig.~\ref{fig:TwoHopIC}). Similarly, for the second hop we have
\begin{align}  
Y_{D1}[t] &=  S^{ \bar{q}- f}  X_{R1}[t]  \label{eq:detTHIChop2D1} \\
Y_{D2}[t] &=  S^{ \bar{q}- f}  X_{R2}[t]  \label{eq:detTHIChop2D2}  \\ 
Y_{S1}[t] &=  S^{\bar{q}- \bar{n}} X_{R1}[t]    \oplus   S^{\bar{q}- \bar{m}} X_{R2}[t]  \label{eq:detTHIChop2S1} \\
 Y_{S2}[t] &=  S^{\bar{q}- \bar{n}}X_{R2}[t]   \oplus   S^{\bar{q}- \bar{m}} X_{R1}[t] \label{eq:detTHIChop2S2}  
\end{align}  
where $X_{R1}[t],X_{R2}[t],Y_{D1}[t],X_{D2}[t],Y_{S1}[t],Y_{S2}[t]$
are binary vectors of length $\bar{q}=\max( \bar{m}, \bar{n}, f)$ and
$S^{ \bar{q}- f}$ etc. are $ \bar{q}\times  \bar{q}$ shift matrices analogously defined. Note
that the relay signals are not only transmitted to their respective
destinations but also overheard by the sources through a backward
interference channel. While the forward channels from the relays to
the destinations have capacity $f$ bits, the backward interference
channel from the relays to the source nodes has capacity $\bar{m}$ and
$\bar{n}$ for the direct and the cross links, respectively.

$X_{R1}[t]$ for $t=1,2,\cdots, N$ is a function of $(Y_{R1}[1], \cdots,
Y_{R1}[t-1])$, while $X_{S1}[t]$ is a function of $( W_1, Y_{S1}[1], \cdots, Y_{S1}[t-1])$ , where $W_1 \in \{1, 2, \cdots, 2^{
  NR_1}\}$ denotes the message of source~$S1$ intended for
destination~$D1$ of rate $R_1$ (similarly for S2, R2 and D2).  We say
$(R_1,R_2)$ is achievable if $W_1$ and $W_2$ can be decoded with
arbitrarily small probability of error at their respective
destinations as $N$ tends to infinity. We define the sum capacity as
$\C_{\text{sum}} = \sup \{ R_{\text{sum}} =R_1 +R_2 : (R_1,R_2)
\,\text{is achievable} \}$.

\section{Main results \label{sec:result}}

The following theorems are the main results of the paper.  
We prove these theorems in   Section~\ref{sec:scheme} and Section~\ref{sec:outerbound}, leaving the details to the appendix.

\vspace{5pt}
\begin{theorem}   [Outer bound]  \label{thm:outerbound} 
For any $f, n, m, \bar{n}, \bar{m},$, the sum capacity of the system defined in Section~\ref{sec:system} is upper bounded by
\begin{align}  \label{eq:outerb}
\C_{\text{sum}}  \leq 
\begin{cases}
     \min \Bigl\{  2\max\{n-m,m \}+2  \max\{ \bar{n},\bar{m} \} ,  &   \\   \quad \quad\quad    2n- m ,  \    2 f   \Bigr\}     &   \!\!\! \!\!\!\!\! \!\!  \!\!\! \!\!\!\!\! \!\!  \text{for }  \alpha \in [0,  2/3] \\
   \min \bigl\{ \  \max\{2n-m,m \} ,  \  2 f \bigr\}   &  \!\!\! \!\!\!\!\! \!\! \!\!\! \!\!\!\!\! \!\!    \text{for  $ \alpha \in [2/3,  2]$} \\
   \min \Bigl\{ n+ f + ( \bar{n} - f )^{+},  & \\ \quad\quad\quad      2n + 2 \bar{n} , \     m,   \  2 f \Bigr\}     &  \!\!\! \!\!\!\!\! \!\! \!\!\! \!\!\!\!\! \!\!    \text{for }  \alpha \in [2,  \infty] 
\end{cases}
\end{align}
%\begin{align}  \label{eq:outerb}
%\C_{\text{sum}}  \leq 
%\begin{cases}
%     \min \Bigl\{  2\max\{n-m,m \}+2  \max\{ \bar{n},\bar{m} \} ,  \   2n- m ,  \    2 f   \Bigr\}     &     \text{for }  \alpha \in [0,  2/3] \\
%   \min \bigl\{ \  \max\{2n-m,m \} ,  \  2 f \bigr\}   &     \text{for  $ \alpha \in [2/3,  2]$} \\
%   \min \Bigl\{ n+ f + ( \bar{n} - f )^{+},  \      2n + 2 \bar{n} , \     m,   \  2 f \Bigr\}     &     \text{for }  \alpha \in [2,  \infty] 
%\end{cases}
%\end{align}
where $\alpha \defeq m/n$, and $(\bullet)^{+} \defeq \max\{\bullet, 0\}$.
\end{theorem}
\vspace{5pt}

\vspace{5pt}
\begin{theorem}   [Inner bound]\label{thm:Achievability} The following  rate is achievable in the system defined in Section~\ref{sec:system} for any $f, n, m, \bar{n}, \bar{m}$: 
\begin{align*}
R_{\text{sum}} =
\begin{cases}
  R_{\text{sum}}^{\text{FBXw}}    &     \text{for }  \alpha \in [0,  2/3]  \\
 R_{\text{sum}}^{\text{NOm}}   &     \text{for   $ \alpha \in [2/3,  2]$}  \\
 R_{\text{sum}}^{\text{RSs}}    &     \text{for }   \alpha \in [2,  \infty] 
\end{cases}
\end{align*}
where
\begin{align}
R_{\text{sum}}^{\text{FBXw}} & \defeq   \min \{   2\max\{ n-m,m \} +  2\bar{m},  \  2n-m , \  2f \}   \label{eq:achieFBXw}   \\
R_{\text{sum}}^{\text{RSs}}  & \defeq    \min \{  \ n + f + ( \bar{n} -f )^{+},     \ 2n+ 2\bar{n},    \  m , \  2f  \}   \label{eq:achieRSs} \\
R_{\text{sum}}^{\text{NOm}}  & \defeq  \min \{  \max\{ 2n - m ,  m \}  ,   2 f \}
\end{align}
\end{theorem}

Combining the outer bound and the inner bound, we immediately get the
following result.

\begin{corollary}   [Sum Capacity]\label{thm:optSCapacity}
  For the system defined in Section~\ref{sec:system}, the inner bound
  and outer bounds match except for the case $(m/n < 2/3, \ 0\leq
  \bar{m} < \bar{n})$. 
\end{corollary}

\begin{remark}
  Note that, when $ f = \infty$, our setting reduces to the setting of
  the two-way interference channel in \cite{SWT:12}, for which the
  capacity also remains open when $(m/n < 2/3, \ 0 \leq \bar{m} <
  \bar{n} )$.
\end{remark}

We also have the following result for the special case when
$\bar{m}=0$, i.e. when there is no cross-link in the backward
interference channel.

\begin{theorem} [Capacity, $\bar{m}=0$] \label{thm:m0} The sum
  capacity of the system defined in Section~\ref{sec:system} when
  $\bar{m} =0$ is given by
\begin{align*}  %\label{eq:outerb}
\C_{\text{sum}}  =
\begin{cases}
     R_{\text{sum}}^{\text{RSw}}      &   \text{for }  \alpha \in [0,  2/3] \\
    R_{\text{sum}}^{\text{NOm}}    &  \text{for  $ \alpha \in [2/3,  2]$} \\
   R_{\text{sum}}^{\text{RSs}}     &  \text{for }  \alpha \in [2,  \infty] 
\end{cases}
\end{align*}
where
\begin{align}
R_{\text{sum}}^{\text{RSw}}   \defeq  \min \Bigl\{  &   f+ \max\{ n-m,m \}  + (  \bar{n} - f )^{+} ,     \nonumber\\ &  2\max\{ n-m,m \} +  2\bar{n},      2n-m ,   2f   \Bigr\}   \label{eq:achieRSw} 
\end{align}
and $R_{\text{sum}}^{\text{NOm}}$ and $R_{\text{sum}}^{\text{RSs}} $ are defined in Theorem~\ref{thm:Achievability}.
\end{theorem}
\vspace{5pt}

$R_{\text{sum}}^{\text{FBXw}}$, $R_{\text{sum}}^{\text{RSw}}$, $R_{\text{sum}}^{\text{RSw}}$ and  $R_{\text{sum}}^{\text{NOm}}$ are the sum rates achieved by various strategies we develop in the next section. Among them, the strategy achieving rate $R_{\text{sum}}^{\text{FBXw}}$ utilizes the backward cross-link for feedback and is optimal in the regime of $(\alpha \in [0,  2/3]$, $\bar{m} \geq \bar{n})$. This strategy is designed to carefully consolidate the two roles of the relay signal so as to provide maximal feedback for the first hop without sacrificing the feedforward capacity over the second hop. Sum rates $R_{\text{sum}}^{\text{RSw}}$ and $R_{\text{sum}}^{\text{RSs}}$ are achieved by feeding back information over the backward direct-link (which leads to rate splitting between feedforward and feedback communication) and are optimal in the regimes of ($ \alpha \in [0,  2/3], \bar{m} =0$) and of ($\alpha \in [2,  \infty] $) respectively.
Finally, $R_{\text{sum}}^{\text{NOm}}$ refers to the no feedback rate which achieves the capacity in the intermediate interference regime ($\alpha \in [ 2/3, 2]$), since feedback is not useful here.

\begin{figure}
\centering
\includegraphics[width=9cm]{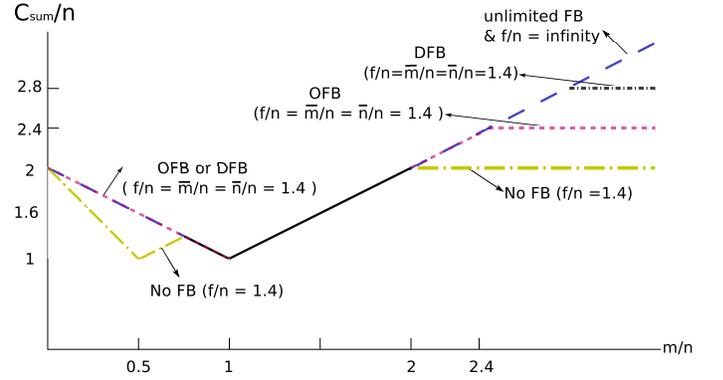}
\caption{Sum capacity comparison between the cases with overheard feedback (OFB), dedicated feedback (DFB), and no feedback (NoFB) for a specific choice of the system parameters.}
\label{fig:FWFB_DeFB_NoFB1}
\end{figure}

\begin{figure}
\centering
\includegraphics[width=9cm]{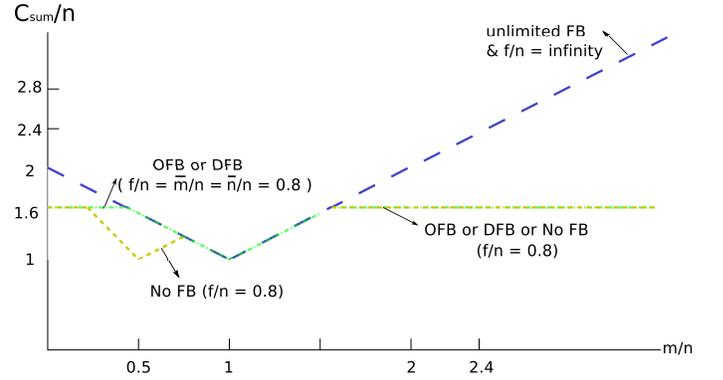}
\caption{Sum capacity comparison between the cases with overheard feedback (OFB), dedicated feedback (DFB), and no feedback (NoFB), for a different set of parameters. The difference between this and the earlier figure is that here $f$, the capacity of the second hop is chosen to be smaller hence in the very weak and strong interference regimes the overall capacity of the system becomes limited by the second hop capacity. In these cases, not surprisingly OFB, DFB and NoFB have the same performance since they only impact the capacity of the first hop.}
\label{fig:FWFB_DeFB_NoFB2}
\end{figure}

Finally, in Figures~\ref{fig:FWFB_DeFB_NoFB1} and \ref{fig:FWFB_DeFB_NoFB2} we compare the sum capacity of our setup to an equivalent setup with dedicated feedback. In this second case, we keep the parameters ($m,n,f,\bar{m}, \bar{n}$) of the system in Fig~\ref{fig:TwoHopIC} the same but assume the forward links over the second hop are decoupled from the backward interference channel which is now dedicated only to feedback. Interestingly, in the regime where ($ \alpha \in [0,  \frac{2}{3}] $ and  $\bar{m}\geq \bar{n}$), overheard-feedback can achieve the same performance as with dedicated-feedback. This is because our achievable strategy is able to consolidate the feedforward and feedback roles of the relay signal without sacrificing the capacity of neither the first or the second hop. On the other hand, in the strong interference regime for the forward channel, i.e. when $ \alpha \in[2,  \infty]$, performance with overheard-feedback suffers a capacity loss compared to the case with dedicated-feedback. Here the two roles of the relay signal are not compatible and as indicated by the upper bound the best strategy is to optimally rate-split between these two roles. However, the optimal rate-splitting strategy is highly non-trivial and has to be carefully designed as detailed in the next section.

\section{Sketches of the achievability proof   }\label{sec:scheme}

In this section we illustrate the strategies we propose in
Theorem~\ref{thm:Achievability} and Theorem~\ref{thm:m0} through specific examples. 
%A detailed description of the strategies can be found in the appendix and \cite{COD:14}.
A detailed description of the strategies can be found in the appendix.

\subsection{Utilizing cross-link overhearing when $\alpha \leq 2/3$  \label{sec:schFBXw}  }

The scheme (named scheme $\Xc_{\text{\text{FBXw}}}$) achieving the rate $R_{\text{sum}}^{\text{FBXw}}$ in
\eqref{eq:achieFBXw} is designed for the case when the first hop is in
the weak interference regime ($\alpha \leq 2/3$) and makes use of the
cross-link feedback for increasing the capacity of the first hop. We
describe the scheme for the special case $m=2, n=4, f=3, \bar{m} =
\bar{n} =1$. The corresponding linear deterministic model for the
first hop is given in Fig.~\ref{fig:Scheme1Example12}, and the model
for the backward interference channel and the second hop is given in
Fig.~\ref{fig:Scheme1Example12FB}. Note that the channel from each
relay to its corresponding destination and the source nodes in Fig.~\ref{fig:Scheme1Example12FB} is a
broadcast channel (the bit levels with the same color correspond to the same signal in this figure).

Our scheme for this case operates in packets of 6 bits per S-D pair. Each packet is
associated with four phases. Among these phases, Phase 1 and Phase 4
involve communication over the first hop while Phase 2 and Phase 3
involve communication over the second hop (and also the backward
interference channel since it represents the overheard signals over
the second hop). See Fig.~\ref{fig:BlocksFWFB}. At the end of these four
phases, the 6 bits of each S-D pair can be decoded by its
corresponding relay. The relaying of these 6 bits to the final
destinations is partially accomplished during the Phases 2 and 3 of
the current packet and the rest is completed during the Phases 2 and 3 of
a future packet.

\begin{figure}
\centering
\includegraphics[width=8.9cm]{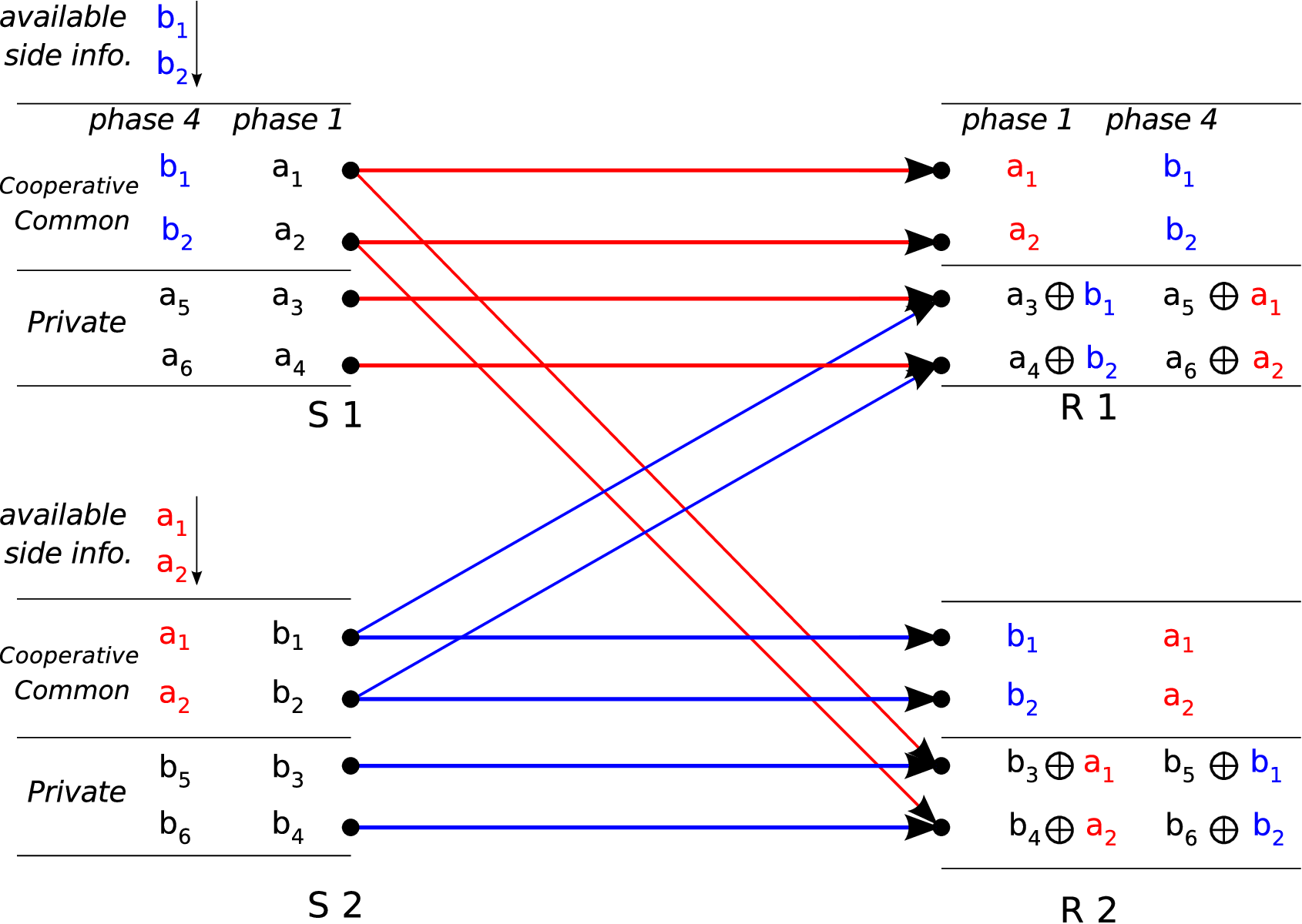}
\caption{The first hop transmission (phase~1 and phase~4) illustration  for scheme $\Xc_{\text{\text{FBXw}}}$ ($m=2, n=4, f=3,  \bar{m} = \bar{n} =1$). }
\label{fig:Scheme1Example12}
\end{figure}

\begin{figure}
\centering
\includegraphics[width=8.9cm]{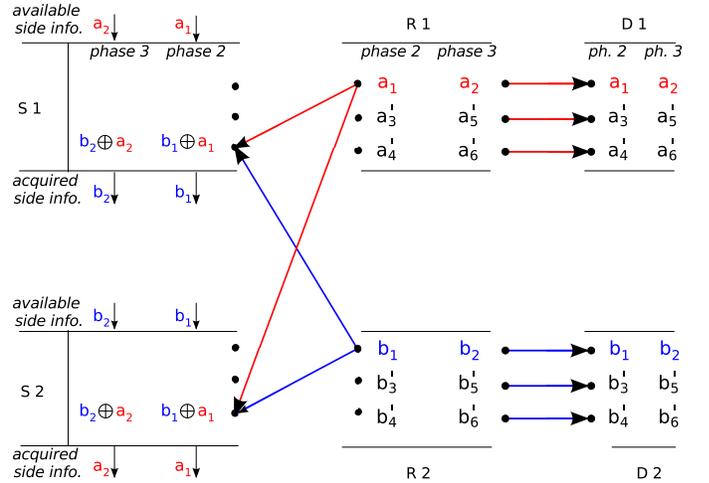}
\caption{The second hop transmission (phase~2 and phase~3) and overhearing illustration  for scheme $\Xc_{\text{\text{FBXw}}}$ ($m=2, n=4, f=3,  \bar{m} = \bar{n} =1$). The bit levels corresponding to the same signal are shown in the same color.} 
\label{fig:Scheme1Example12FB}
\end{figure}

Let $\{a_t \}_{t=1}^{6}$ and $\{b_t\}_{t=1}^{6}$ denote the $6$ bits
of S1 and S2 destined to D1 and D2 respectively. We next describe the
four phases associated with this packet.

\textit{Phase 1:} Each source sends two cooperative common bits (over the upper
two signal levels) and two private bits (over the lower two signal
levels). The cooperative common bits $a_1, a_2$ of S1 and $b_1, b_2$ of S2 are
received interference free at the corresponding relays, while the
private bits, $a_3, a_4$ of S1 and $b_3, b_4$ of S2, which are only
visible to the corresponding relays, arrive interfered with the common
bits of the other source node. ($a_1, a_2$ are called common bits in the sense that they can be decoded by both relays. They are called cooperative in the sense that they are known/learned by both source nodes. Private bits are known only by the corresponding source and relay nodes.) See Fig.~\ref{fig:Scheme1Example12}. In
the next two phases, each source node will learn the interfering
(common) bits of the other source node, which will allow to resolve
the interference in the fourth phase.

\textit{Phases 2 and 3:} In the beginning of these phases the relays
have recovered the cooperative common bits of their respective source nodes. These
common bits need to be both forwarded to their destination and fed
back to the source nodes to resolve interference. We use the upper
most level of the relay signals to accomplish both of these goals at
the same time. R1 transmits $a_1$ and $a_2$ over its uppermost signal
level in Phase 2 and Phase 3 respectively, while R2 transmits
similarly $b_1$ and $b_2$. Since S1 already knows $a_1$ and $a_2$ and
S2 already knows $b_1$ and $b_2$, S1 can decode $b_1$ and $b_2$ and S2
can decode $a_1$ and $a_2$ from the linear combinations $a_1\oplus
b_1$ and $a_2\oplus b_2$ which they gather during these two
phases. Note that we fully exploit the common uppermost bit level by
using it both for forward and backward information transmission. The
remaining two bit levels of the forward channel are utilized to send
the four bits $a'_3, a'_4, a'_5, a'_6$ of S1 and $b'_3,
b'_4, b'_5, b'_6$ of S2 from the packet before last. These are
the bits of an earlier packet (say, the packet before the last one), that have been decoded by the relays
but have not been forwarded to the destinations yet.

\textit{Phase 4:} Having learned the cooperative common bits of S2, S1 transmits $b_1, b_2$ and two fresh  bits  $a_5, a_6$ from its upper and lower signal levels respectively,  while S2 sends  $a_1, a_2$ and two fresh bits  $b_5, b_6$ in a similar way. Note that this allows R1 to resolve the interference and recover $a_3, a_4$ from Phase 1 and $a_5, a_6$ from the current phase, and similarly for R2. Note that while R1 and R2 can decode $\{a_t\}_{t=3}^6$ and $\{b_t\}_{t=3}^6$ respectively, these bits have not been communicated yet to their final destinations. This is accomplished in the Phases 2 and 3 corresponding to a future packet (say, the packet after next), similarly to $\{a'_t\}_{t=3}^6$ and $\{b'_t\}_{t=3}^6$  which were transmitted in the Phases 2 and 3 corresponding to the current packet. See the timeline in Fig.~\ref{fig:BlocksFWFB}.

\medbreak

With this strategy $12$ bits can be communicated to their destinations in every 2 channel uses (assuming a large number of packets and ignoring the effect of the first two and the last two packets). This yields a sum rate of $6$ bits per channel use, which turns out to be optimal for this channel. See Theorems~\ref{thm:outerbound} and \ref{thm:Achievability}.

\begin{figure}
\centering
\includegraphics[width=8.9cm]{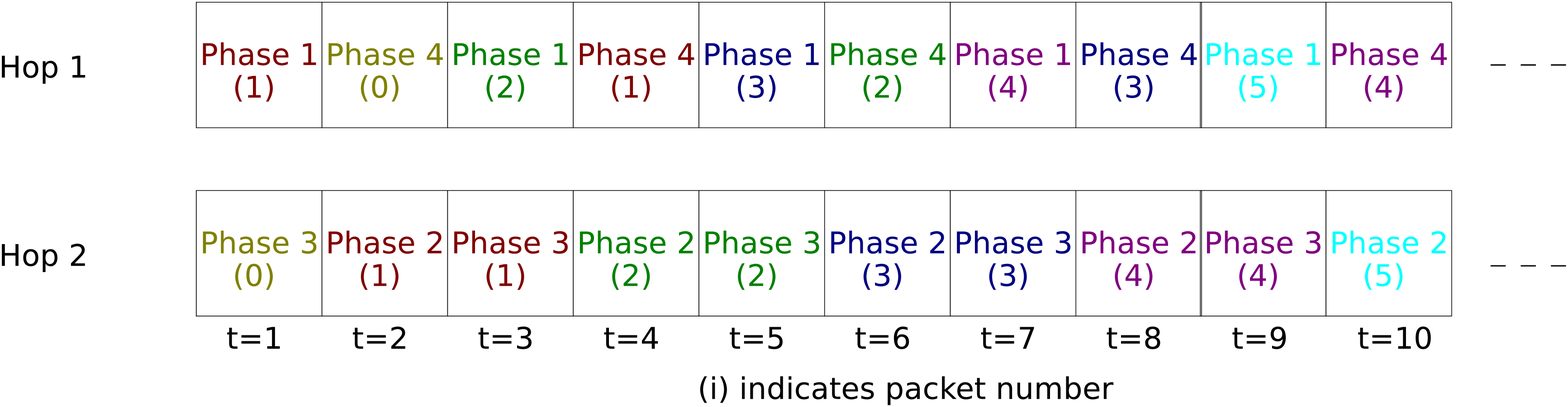}
\caption{Schematic of phases and packets: For packet~$i$, $i=1,2,\cdots$, phase~1 (at time~$t=2i-1$) and phase~4 (at time~$t=2i+2$) correspond to hop~1 transmission, while phase~2 (at time~$t=2i$) and phase~3 (at time~$t=2i+1$) correspond to hop~2 transmission. Different colors refer to different packets.}
\label{fig:BlocksFWFB}
\end{figure}

%%%%%%%%%%%%%%%%%%%%

\subsection{Utilizing direct-link overhearing when $(\alpha \leq 2/3, \bar{m}  = 0) $ \label{sec:schRSw} }

The scheme  (named scheme $\Xc_{\text{\text{RSw}}}$) achieving the rate $R_{\text{sum}}^{\text{RSw}}$ in
\eqref{eq:achieRSw} is designed for the case when the first hop is in
the weak interference regime and there is no cross-link in the backward
interference channel ($\alpha \leq 2/3, \bar{m} = 0 $). In contrast to
the previous scheme, this scheme makes use of the direct-link feedback
from R1 to S1 and from R2 to S2 for increasing the capacity of the first hop. In this case there is a tradeoff between utilizing the relay bits for conveying information to
the destinations versus feeding back information to the source nodes,
which results in rate splitting. We describe the scheme for the
special case $m\!=\!2, n\!=\!4, f\!=\!3, \bar{n} \!=\!1, \bar{m} \!=\!
0 $. The corresponding linear deterministic models for the channels in the first and the second hops are given in Fig.~\ref{fig:Scheme3EgPh1} and Fig.~\ref{fig:Scheme3EgFB} respectively.

As before, this scheme operates in packets, in this case of $5$ bits per S-D pair and each packet is
associated with four phases, \emph{i.e.,} Phase 1 and Phase 4 involve
communication over the first hop while Phase 2 and Phase 3 involve
communication over the second hop and the backward channel. See
Fig.~\ref{fig:BlocksFWFB}. 
At the end of these four phases, the $5$ bits of each S-D pair can be
decoded by its corresponding relay. The relaying of these $5$ bits to
the final destination will again be accomplished partially in the
Phases 2 and 3 of the current packet, and the rest in Phases 2 and 3
of the packet after next.
Let $\{a_t \}_{t=1}^{5}$ and $\{b_t\}_{t=1}^{5}$ denote the $5$ bits
of S1 and S2 destined to D1 and D2 respectively. We next describe the
four phases associated with this packet.

\begin{figure}
\centering
\includegraphics[width=8.9cm]{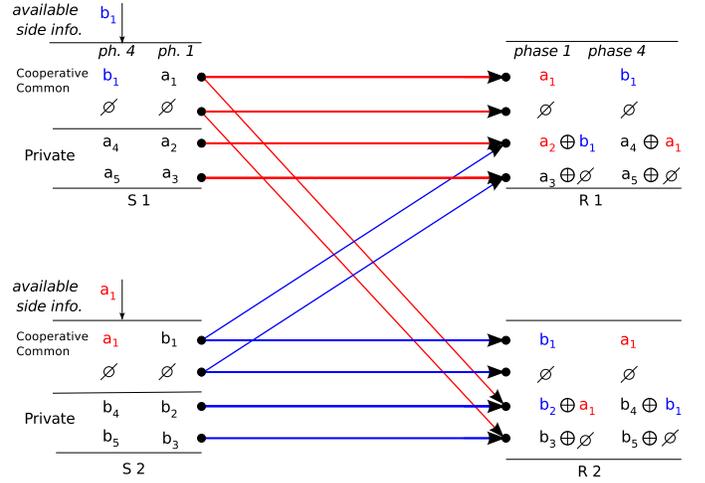}
\caption{The first hop transmission (phase~1 and phase~4) illustration  for scheme $\Xc_{\text{\text{RSw}}}$ ($m=2, n=4, f=3,  \bar{n} =1, \bar{m} = 0$). }
\label{fig:Scheme3EgPh1}
\end{figure}

\begin{figure}
\centering
\includegraphics[width=8.9cm]{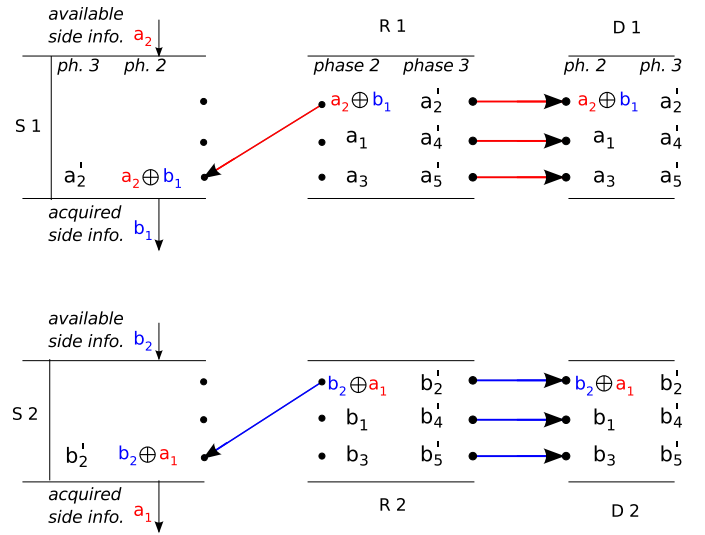}
\caption{The second hop transmission (phase~2 and phase~3) and overhearing illustration  for scheme $\Xc_{\text{\text{RSw}}}$ ($m=2, n=4, f=3,  \bar{n} =1, \bar{m} =0 $). } 
\label{fig:Scheme3EgFB}
\end{figure}

\textit{Phase 1:} Each source sends one cooperative common bit (over its uppermost signal level) and two private bits (over its lowest two  signal levels). The second bit level is not utilized and is fixed to $0$. The most significant bit, $a_1$ of S1 and $b_1$ of S2,  as well as the least significant bit, $a_3$ of S1 and $b_3$ of S2,  are received interference free at the corresponding relays, while the remaining private bit, $a_2$ of S1 and $b_2$ of S2,  arrives interfered with the common bit of the other source node. See Fig.~\ref{fig:Scheme3EgPh1}. In the next two phases, each source node will learn the interfering (common) bit of the other source node, which will allow to resolve the interference in the fourth phase.

\textit{Phase 2 and 3:} In the beginning of  Phase~2, each relay has recovered the cooperative common bit and one private bit of its respective source node, and these bits need to be forwarded to their destination. As before, each source needs to learn the common information of the other source node in order to resolve interference. However, since $\bar{m}=0$ this information needs to be  fed back  over the backward direct-link rather than the backward cross-link. In phase 2, R1 and R2 feed back $a_2\oplus b_1$ and  $b_2\oplus a_1$, received in the previous phase,  to S1 and S2 respectively through the backward  direct-link (uppermost signal level). Then the common bits $b_1$ and $a_1$ can be decoded by S1 and S2 respectively by using the side information at each source node ($a_1$ of S1, $b_1$ of S2). On the remaining lower two signal levels  R1 transmits $ a_1$ and $a_3$, while R2 transmits $b_1$ and $b_3$. During Phase~3, the relay signal is fully utilized for forward information transmission. R1 sends the three bits $a'_2,  a'_4, a'_5$ of S1, and R2 sends the bits $b'_2,  b'_4, b'_5$ of S2, which belong to an earlier packet. Note that in this case the relay bits are split between feedforward and feedback communication. In Phase 2, the uppermost bit level of the relay is used solely for feedback and does not provide any useful information to the corresponding destination. Similarly in Phase 3, this bit level is used for sending information to the destination and does not provide any useful feedback information for the source node.

\textit{Phase 4:} Having learned the cooperative common bit of the other source node, S1 transmits $b_1$ and two fresh  bits  $a_4, a_5$ from its upper and lower signal levels respectively (the second bit level is again not utilized),  while S2 sends  $a_1$ and two fresh bits  $b_4, b_5$ in a similar way. Note that this allows R1 to resolve the interference and recover $a_2$ from Phase 1 and $a_4, a_5$ from the current phase, and similarly for R2. Note that while R1 and R2 can decode $\{a_2,a_4,a_5\}$ and $\{b_2,b_4,b_5\}$ respectively, these bits have not been communicated yet to their final destinations. This is accomplished in the Phase 3 corresponding to the packet after next.

\medbreak

With this strategy $10$ bits are communicated to their destinations in every 2 channel uses. This yields a sum rate of $5$ bits per channel use which is optimal for this channel. See Theorem~\ref{thm:m0}.

\subsection{Utilizing direct-link overhearing when  $\alpha \geq 2$ \label{sec:schRSs} }

\begin{figure}
\centering
\includegraphics[width=8.9cm]{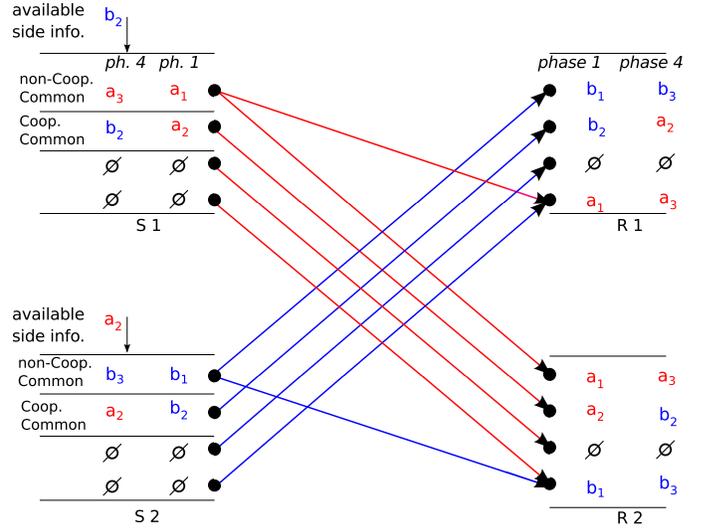}
\caption{The first hop transmission (phase~1 and phase~4) illustration  for scheme $\Xc_{\text{\text{RSs}}}$ ($m=4, n=1, f=2,  \bar{n} =1, \bar{m} = 1$). }
\label{fig:Scheme2EgPh1}
\end{figure}

\begin{figure}
\centering
\includegraphics[width=8.9cm]{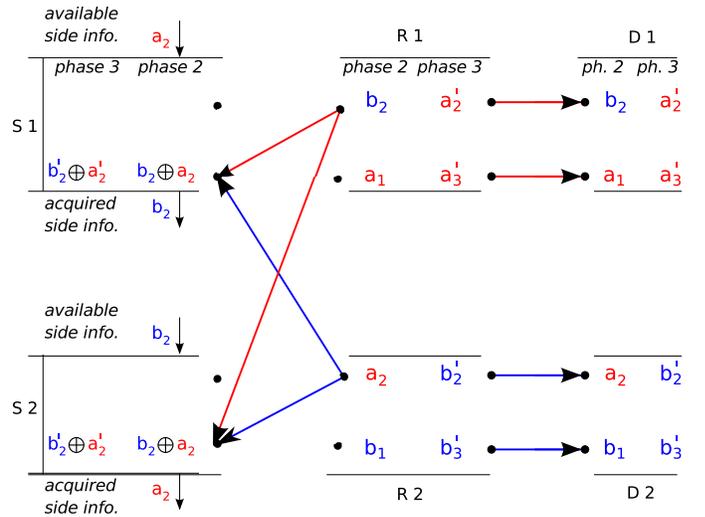}
\caption{The second hop transmission (phase~2 and phase~3) and overhearing illustration  for scheme $\Xc_{\text{\text{RSs}}}$ ($m=4, n=1, f=2,  \bar{n} =1, \bar{m} = 1$). } 
\label{fig:Scheme2EgFB}
\end{figure}

The scheme  (named scheme $\Xc_{\text{\text{RSs}}}$) achieving the rate $R_{\text{sum}}^{\text{RSs}}$ in \eqref{eq:achieRSs} is designed for the case when the first hop is in the strong interference regime  ($\alpha \geq 2$). 
The scheme is briefly illustrated in Fig.~\ref{fig:Scheme2EgPh1} and Fig.~\ref{fig:Scheme2EgFB} for the special case $m=4, n=1, f=2,  \bar{n} =1, \bar{m} = 1$. 
Similarly to the previous scheme, this four-phase scheme is based on direct-link overhearing and it uses rate splitting at the relay for feedback and feedforward to achieve the optimal sum rate $3$ bits/channel use. The difference between this scheme and the earlier one is that the while in the earlier scheme S1 uses its side information about the cooperative bit of S2 to resolve the interference at R1 in the next phase, here S1 uses its side information to enhance the communication between S2 and R2 by communicating the cooperative bit to R2. This is similarly the case for the interference channel with standard dedicated feedback.

As before, this scheme operates in packets, in this case of $3$ bits per S-D pair, and each packet is
associated with four phases. See
Fig.~\ref{fig:BlocksFWFB}. 
At the end of these four phases, the $3$ bits of each S-D pair can be
decoded by its corresponding relay. The relaying of these $3$ bits to
the final destination will again be accomplished partially in the
Phases 2 and 3 of the current packet, and the rest in Phases 2 and 3
of the packet after next.
Let $\{a_t \}_{t=1}^{3}$ and $\{b_t\}_{t=1}^{3}$ denote the $3$ bits
of S1 and S2 destined to D1 and D2 respectively. We next describe the
four phases associated with this packet.

\textit{Phase 1:} Each source sends one non-cooperative common bit over its uppermost signal level and one cooperative common bit over its  second upper most signal level. The rest of the levels are not utilized and are fixed to $0$. The non-cooperative common bit, $a_1$ of S1 and $b_1$ of S2,  is received interference free at both relays, while the cooperative common bit, $a_2$ of S1 and $b_2$ of S2,  is received (interference free) at the unintended relays, i.e. R2 and R1 respectively. See Fig.~\ref{fig:Scheme2EgPh1}.  In order to improve the first hop capacity, in the next two phases the  cooperative common bit  $a_2$ will be fed back from R2 to S2 such that it can be sent from S2 to R1 in the fourth phase,  while the cooperative common bit  $b_2$ will be fed back from R1 to S1 such that it can be sent from S1 to R2 in the fourth phase.

\textit{Phase 2 and 3:} At the beginning of  Phase~2, each relay has recovered one non-cooperative common bit of its corresponding  source and one cooperative common bit of the other source node. 
During Phase~2, R1 and R2 feed back the cooperative common bits $ b_2$ and  $a_2$  to S1 and S2 respectively through the backward  direct-link (uppermost signal level). Then the cooperative common bits $b_2$ and $a_2$ can be decoded by S1 and S2 respectively by using the available side information ($a_2$ of S1 and $b_2$ of S2). 
The remaining signal level of the relay signal is utilized to forward the non-cooperative common bits,  $a_1$ of S1 and $b_1$ of S2, to its destination.   
During Phase~3, the relay signal is fully utilized for forward information transmission. R1 sends the two bits $a'_2,  a'_3$ of S1, and R2 sends the bits $b'_2,  b'_3$ of S2, which belong to the packet before last.  Note that in this case the relay bits are split between feedforward and feedback communication. In Phase 2, the uppermost bit level of the relay is used solely for feedback and does not provide any useful information to the corresponding destination. Similarly in Phase 3, this bit level is used for sending information to the destination and does not provide any useful feedback information for the source node.

\textit{Phase 4:} Having learned the cooperative common bit of the other source node, S1 transmits $b_2$ and fresh  bit $a_3$ from its upper second and uppermost signal levels respectively,  while S2 sends  $a_2$ and fresh bit  $b_3$ in a similar way. Note that this allows R1 to learn  $a_2$ via the loop $S1 \rightarrow R2 \rightarrow S2 \rightarrow R1$, and allows R2 to learn $b_2$ via the loop $S2 \rightarrow R1 \rightarrow S1 \rightarrow R2$. Note that while R1 and R2 can decode $\{a_2, a_3\}$ and $\{b_2,b_3\}$ respectively, these bits have not been communicated yet to their final destinations. This is accomplished in the Phase 3 corresponding to the packet after next.

\medbreak

With this strategy $6$ bits are communicated to their destinations in every 2 channel uses. This yields a sum rate of $3$ bits per channel use which is optimal for this channel. See Theorems~\ref{thm:outerbound} and \ref{thm:Achievability}.

\section{Proof of the Outer Bound   \label{sec:outerbound}}

It suffices to  prove the general bounds
\begin{align}    
\C_{\text{sum}} \leq    \min  \bigl\{  &   n +f + ( \bar{n} -f )^{+},   \label{eq:GeneralOB1}   \\   &    2n+ 2\bar{n},    \label{eq:GeneralOB2} \\  &     2\max\{n-m,m \}+2 \max\{ \bar{n},\bar{m} \},  \label{eq:GeneralOB3}  \\ &  \max\{n,m\}+ (n-m)^{+},  \label{eq:GeneralOB4}  \\ &  2f \label{eq:GeneralOB5}   \bigr\}   
\end{align}
for any $f, n, m, \bar{n}, \bar{m}$, and the special bound
\begin{align}    
\C_{\text{sum}} \leq   f+  \max\{ n-m,  m \} + ( \bar{n} -f )^{+}  \  \text{for} \  \bar{m} =0. \label{eq:GeneralOB6}       
\end{align}
Note that the outer bound in Theorem~\ref{thm:outerbound} follows by evaluating the outer bounds \eqref{eq:GeneralOB2}-\eqref{eq:GeneralOB6} for different regimes of $\alpha$. 

The outer bound \eqref{eq:GeneralOB5} follows by observing that the sum capacity of the whole system can not be more than the second hop sum capacity $2f$. The outer bounds \eqref{eq:GeneralOB2}, \eqref{eq:GeneralOB3}, \eqref{eq:GeneralOB4} are the bounds on the sum capacity of the interference channel in the first hop   from \cite{SWT:12, VSA:12} with dedicated feedback which consequently serve as outer bounds for our case with overheard feedback. In the following, we concentrate on proving the outer bounds \eqref{eq:GeneralOB1} and \eqref{eq:GeneralOB6}.  For notational convenience, we let $ Y^{N}_{Q} \defeq \{Y_{Q}[t]\}_{t=1}^{N}$ for $Q \in \{R1,R2,S1,S2,D1,D2\}$; and let $ X^{N}_{Q} \defeq \{X_{Q}[t]\}_{t=1}^{N}$, $X_{Q,1:\tau}[t] \defeq   \Bmatrix{ X_{Q,1}[t], \cdots, X_{Q,\tau}[t] }^\T $, $ X^{N}_{Q,1:\tau} \defeq \{X_{Q,1:\tau}[t]\}_{t=1}^{N}$ for $Q \in \{R1,R2,S1,S2\}$, $\tau \in \{m,n,\bar{m},\bar{n},f\}$. Note that, $X_{S1,1:n}[t]$ is the part of $X_{S1}[t]$ visible to R1, while $X_{S1,1:m}[t]$ is visible to R2.

\subsection{Proof of outer bound  \eqref{eq:GeneralOB1} }

Starting with Fano's inequality, we have
\begin{align}
& N (R_1+ R_2-   \epsilon_N) \nonumber\\
&\leq I(W_1; Y^{N}_{D1})  +  I(W_2; Y^{N}_{D2}) \label{eq:fano302}\\
&\leq I(W_1; Y^{N}_{R1})  +  I(W_2; Y^{N}_{D2}) \label{eq:MarkovYS}\\
&\leq I(W_1; Y^{N}_{R1}, W_2)  +  I(W_2; Y^{N}_{D2})  \label{eq:entropyadd}\\
&= I(W_1; Y^{N}_{R1}| W_2) +  \underbrace{ I(W_1; W_2)}_{0}   +  I(W_2; Y^{N}_{D2})     \nonumber\\
&= I(W_1; Y^{N}_{R1}| W_2)  +  I(W_2; Y^{N}_{D2})  \label{eq:W1W2ind}\\
&= H(Y^{N}_{R1}| W_2) \! - \! \underbrace{  H(Y^{N}_{R1}| W_2, W_1)}_{0}  \! +\!   H( Y^{N}_{D2}) \! -\! H(Y^{N}_{D2}|W_2)    \nonumber\\
&= H(Y^{N}_{R1}| W_2)  +  H( Y^{N}_{D2})  -H(Y^{N}_{D2}|W_2)   \label{eq:W1W2rec} \\
&= H(Y^{N}_{R1}| W_2)  +  H( Y^{N}_{D2})   \nonumber\\& \quad  -  \bigl(  H(Y^{N}_{D2}, Y^{N}_{R1}|W_2) -H(Y^{N}_{R1}|W_2,Y^{N}_{D2} )  \bigr)  \nonumber\\
&= H(Y^{N}_{R1}| W_2)  +  H( Y^{N}_{D2})  \nonumber\\& \   \   - \!  \bigl(   \!  H(Y^{N}_{R1}|W_2) + H(Y^{N}_{D2}|W_2,Y^{N}_{R1}) -H(Y^{N}_{R1}|W_2,Y^{N}_{D2} )   \bigr)  \nonumber\\
&=   H( Y^{N}_{D2})  + H(Y^{N}_{R1}|W_2,Y^{N}_{D2})-    \underbrace{H(Y^{N}_{D2}|W_2,Y^{N}_{R1})}_{\geq H(Y^{N}_{D2}|W_2,Y^{N}_{R1}, W_1) =0}   \nonumber\\
&\leq   H( Y^{N}_{D2})  + H(Y^{N}_{R1}|W_2,Y^{N}_{D2})  \label{eq:ConditionRe}   \\
&\leq   H( Y^{N}_{D2})  + H(Y^{N}_{R1}, X^{N}_{R2,1:\bar{n}}|W_2,Y^{N}_{D2})  \label{eq:addEntropy}   \\
&=   H( Y^{N}_{D2})  + H(X^{N}_{R2,1:\bar{n}}|W_2,Y^{N}_{D2})  \nonumber\\& \quad + H(Y^{N}_{R1}|W_2,Y^{N}_{D2},X^{N}_{R2,1:\bar{n}})  \label{eq:chainrule235}   \\
&=   H( Y^{N}_{D2})    + H(X^{N}_{R2,1:\bar{n}}|W_2,Y^{N}_{D2}) \nonumber\\& \quad  + \sum H(Y_{R1}[t]|W_2,Y^{N}_{D2},X^{N}_{R2,1:\bar{n}},Y^{t-1}_{R1})  \label{eq:ChainRule}   \\
&=   H( Y^{N}_{D2})  + H(X^{N}_{R2,1:\bar{n}}|W_2,Y^{N}_{D2}) \nonumber\\& \quad   +  \sum H(Y_{R1}[t]|W_2,Y^{N}_{D2},X^{N}_{R2,1:\bar{n}},Y^{t-1}_{R1}, X^{t}_{R1})  \label{eq:fbfunction1}   \\
&=   H( Y^{N}_{D2})  + H(X^{N}_{R2,1:\bar{n}}|W_2,Y^{N}_{D2}) \nonumber\\& \quad  +  \sum H(Y_{R1}[t]|W_2,Y^{N}_{D2},X^{N}_{R2,1:\bar{n}},Y^{t-1}_{R1}, X^{t}_{R1}, Y^{t}_{S2})  \label{eq:fbfunction21}   \\
&=   H( Y^{N}_{D2})  + H(X^{N}_{R2,1:\bar{n}}|W_2,Y^{N}_{D2})  \nonumber\\& \   +  \! \sum \!H(Y_{R1}[t]|W_2,Y^{N}_{D2},X^{N}_{R2,1:\bar{n}},Y^{t-1}_{R1}, X^{t}_{R1}, Y^{t}_{S2}, X^{t+1}_{S2})  \label{eq:fbfunction43}   \\
&=   H( Y^{N}_{D2})  + H(X^{N}_{R2,1:\bar{n}}|W_2,Y^{N}_{D2})  \nonumber\\& \  + \! \!\sum \! H(\!  X_{S1,1:n}[t]|W_2,Y^{N}_{D2}, \! X^{N}_{R2,1:\bar{n}},  Y^{t-1}_{R1}, \! X^{t}_{R1}, \! Y^{t}_{S2}, \! X^{t+1}_{S2} \!)  \label{eq:fbfunction57}   \\
&\leq   H( Y^{N}_{D2})  + H(X^{N}_{R2,1:\bar{n}}|Y^{N}_{D2}) +  \sum H(X_{S1,1:n}[t])  \label{eq:ConditionRe213}   \\
&=   H( Y^{N}_{D2})  + H(X^{N}_{R2,1:\bar{n}}| X^{N}_{R2,1:f}) +  \sum H(X_{S1,1:n}[t])   \\
&\leq  N f  +N  ( \bar{n} -f )^{+}+  N n  \label{eq:sumratefinal}  
\end{align}
where $\epsilon_N \rightarrow 0$ when $N\rightarrow \infty$; %\eqref{eq:fano302} stems from Fano's inequality;  
\eqref{eq:MarkovYS} follows from the fact that $W_1 \rightarrow Y^{N}_{R1} \rightarrow  Y^{N}_{D1}$ forms a Markov chain ($Y^{N}_{D1}$ is a function of $Y^{N}_{R1}$);
\eqref{eq:entropyadd} results from the fact that adding information increases the mutual information;
\eqref{eq:W1W2ind} follows from the independence of the two messages $W_1 $ and $W_2$;
\eqref{eq:W1W2rec} uses  the fact that the knowledge of $\{W_2, W_1\}$ implies the knowledge of $Y^{N}_{R1}$;
\eqref{eq:ConditionRe}  follows from the fact that conditioning reduces entropy, and the fact that the knowledge of $\{W_2, W_1\}$ implies the knowledge of $Y^{N}_{D2}$;
\eqref{eq:addEntropy} follows from the fact that adding information increases the entropy;
\eqref{eq:ChainRule} folows from the chain rule for entropy;
\eqref{eq:fbfunction1} uses the fact that $X_{R1}[t]$ is a function of $Y^{t-1}_{R1}$;
\eqref{eq:fbfunction21} follows from the fact that $Y_{S2}[t]$ is a function of $\{X_{R2,1:\bar{n}}[t],X_{R1,1:\bar{m}}[t] \}$;
\eqref{eq:fbfunction43} results from the fact that $X_{S2}[t]$ is a function of $\{W_2,Y^{t-1}_{S2} \}$;
\eqref{eq:fbfunction57} uses the fact that $Y_{R1}[t]$ is a deterministic function of $\{X_{S1,1:n}[t], X_{S2,1:m}[t] \}$;
\eqref{eq:ConditionRe213}  follows from the fact that conditioning reduces entropy.
From \eqref{eq:sumratefinal}, we finally have the bound $R_1+ R_2\leq n+f+  ( \bar{n} -f )^{+}$ when  $N$ is large.

\subsection{Proof of outer bound  \eqref{eq:GeneralOB6}}

The  outer bound  \eqref{eq:GeneralOB6} is valid for the case without feedback cross-link ($\bar{m} =0$).
In this case, we have
\begin{align}
& N R_1+ N R_2-   N\epsilon_N \nonumber\\
&=   H(W_1)  +  H(W_2)  -  N \epsilon_N \nonumber\\
&\leq I(W_1; Y^{N}_{D1})  +  I(W_2; Y^{N}_{D2}) \label{eq:2fano302}\\
&\leq I(W_1; Y^{N}_{R1})  +  I(W_2; Y^{N}_{R2}) \label{eq:2MarkovYS}\\
& =  H(Y^{N}_{R1}) - H( Y^{N}_{R1}| W_1)  +  H( Y^{N}_{R2}) - H(Y^{N}_{R2}|W_2)  \nonumber\\
& =  H(Y^{N}_{R1}) - H( X^{N}_{S2,1:m}| W_1)  +  H( Y^{N}_{R2}) \!-\! H(X^{N}_{S1,1:m}|W_2)  \label{eq:equal236}  \\
& =  H(Y^{N}_{R1}) - \bigl(  H( X^{N}_{S2,1:m})   -I( X^{N}_{S2,1:m}; W_1) \bigr)   \nonumber\\& \quad +  H( Y^{N}_{R2}) - \bigl( H(X^{N}_{S1,1:m})  -  I(X^{N}_{S1,1:m};W_2)    \bigr)   \nonumber\\ 
& \leq  H(Y^{N}_{R1}, X^{N}_{S1,1:m}) - \bigl(  H( X^{N}_{S2,1:m})   -  I( X^{N}_{S2,1:m}; W_1) \bigr)   \nonumber\\& \quad   +  H( Y^{N}_{R2}, X^{N}_{S2,1:m}) -   \bigl( H(X^{N}_{S1,1:m})  -  I(X^{N}_{S1,1:m};W_2)    \bigr)    \label{eq:addentropy}  \\ 
& =  H(Y^{N}_{R1}, X^{N}_{S1,1:m})  -  H(X^{N}_{S1,1:m})  +  I(X^{N}_{S1,1:m};W_2)    \nonumber\\& \quad  +  H( Y^{N}_{R2}, X^{N}_{S2,1:m}) -  H( X^{N}_{S2,1:m})   + I( X^{N}_{S2,1:m}; W_1)   \nonumber\\ 
& =  H(Y^{N}_{R1}| X^{N}_{S1,1:m})   +  I(X^{N}_{S1,1:m};W_2)    \nonumber\\& \quad  +  H( Y^{N}_{R2}|X^{N}_{S2,1:m})  + I( X^{N}_{S2,1:m}; W_1)  \label{eq:2ChainRule} \\ 
& \leq H(Y^{N}_{R1}| X^{N}_{S1,1:m})   +    I(X^{N}_{S1,1:m}, X^{N}_{R1,1:\bar{n}}, Y^{N}_{D1}, W_1;W_2)   \nonumber\\& \quad   + H( Y^{N}_{R2}|X^{N}_{S2,1:m})  + I( X^{N}_{S2,1:m}, X^{N}_{R2,1:\bar{n}}, Y^{N}_{D2}, W_2; W_1)  \label{eq:2entropyadd} \\ 
& =   I( X^{N}_{R1,1:\bar{n}},Y^{N}_{D1}, W_1;W_2)  + I( X^{N}_{R2,1:\bar{n}}, Y^{N}_{D2}, W_2; W_1)  \nonumber\\& \quad +H(Y^{N}_{R1}| X^{N}_{S1,1:m})   + H( Y^{N}_{R2}|X^{N}_{S2,1:m})  \label{eq:function326}  \\ 
& =   I(Y^{N}_{D1}, W_1;W_2) +I( X^{N}_{R1,1:\bar{n}};W_2|Y^{N}_{D1}, W_1)   \nonumber\\& \quad  + I(  Y^{N}_{D2}, W_2; W_1) + I( X^{N}_{R2,1:\bar{n}}; W_1| Y^{N}_{D2}, W_2) \nonumber\\ &\quad  +H(Y^{N}_{R1}| X^{N}_{S1,1:m})   + H( Y^{N}_{R2}|X^{N}_{S2,1:m})  \label{eq:function9837}  \\ 
& =   I(Y^{N}_{D1}, W_1;W_2)  + I(  Y^{N}_{D2}, W_2; W_1)  \nonumber\\& \quad  +H(Y^{N}_{R1}| X^{N}_{S1,1:m})  + H( Y^{N}_{R2}|X^{N}_{S2,1:m})    \nonumber\\ &\quad  + \underbrace{ H( X^{N}_{R1,1:\bar{n}}|Y^{N}_{D1}, W_1)}_{ \leq H( X^{N}_{R1,1:\bar{n}}|Y^{N}_{D1})  } -  \underbrace{H( X^{N}_{R1,1:\bar{n}}|Y^{N}_{D1}, W_1,W_2)}_{=0}  \nonumber\\ &\quad  +\underbrace{H( X^{N}_{R2,1:\bar{n}}| Y^{N}_{D2}, W_2)}_{\leq H( X^{N}_{R2,1:\bar{n}}| Y^{N}_{D2})} - \underbrace{H( X^{N}_{R2,1:\bar{n}}| Y^{N}_{D2}, W_2,W_1)}_{=0} \label{eq:ChainRule3837}  \\
&  \leq   I(Y^{N}_{D1}, W_1;W_2)  + I(  Y^{N}_{D2}, W_2; W_1)    \nonumber\\ &\quad   +H(Y^{N}_{R1}| X^{N}_{S1,1:m})   + H( Y^{N}_{R2}|X^{N}_{S2,1:m})   \nonumber\\ &\quad   + H( X^{N}_{R1,1:\bar{n}}|Y^{N}_{D1})  +  H( X^{N}_{R2,1:\bar{n}}| Y^{N}_{D2}) \label{eq:W1W2recon332}  \\
& =   H(Y^{N}_{D1},W_1)  - \underbrace{H(W_1|W_2)}_{=H(W_1)} - \underbrace{H(Y^{N}_{D1}|W_1,W_2)}_{ =0 }    \nonumber\\ &\quad  + H(Y^{N}_{D2},W_2)  - \underbrace{H(W_2|W_1)}_{=H(W_2)} - \underbrace{H(Y^{N}_{D2}|W_1,W_2)}_{=0 }   \nonumber\\ &\quad   +H(Y^{N}_{R1}| X^{N}_{S1,1:m})   + H( Y^{N}_{R2}|X^{N}_{S2,1:m})  \nonumber\\ &\quad  + H( X^{N}_{R1,1:\bar{n}}|Y^{N}_{D1})  +  H( X^{N}_{R2,1:\bar{n}}| Y^{N}_{D2}) \label{eq:2ChainRule564}  \\ 
& =   H(Y^{N}_{D1},W_1)  - H(W_1)    + H(Y^{N}_{D2},W_2)  - H(W_2)  \nonumber\\ &\quad +H(Y^{N}_{R1}| X^{N}_{S1,1:m})   + H( Y^{N}_{R2}|X^{N}_{S2,1:m}) \nonumber\\ &\quad + H( X^{N}_{R1,1:\bar{n}}|Y^{N}_{D1})  +  H( X^{N}_{R2,1:\bar{n}}| Y^{N}_{D2}) \label{eq:2W1W2ind}  \\ 
& =   H(Y^{N}_{D1}) + \underbrace{H(W_1|Y^{N}_{D1})}_{ \leq N \epsilon_N }  -  \underbrace{H(W_1)}_{=NR_1}   \nonumber\\ &\quad   + H(Y^{N}_{D2})+ \underbrace{H(W_2|Y^{N}_{D2}) }_{ \leq N \epsilon_N } -\underbrace{H(W_2)}_{=N R_2} \nonumber\\ &\quad +H(Y^{N}_{R1}| X^{N}_{S1,1:m})  +H( Y^{N}_{R2}|X^{N}_{S2,1:m})   \nonumber\\ &\quad    + H( X^{N}_{R1,1:\bar{n}}|Y^{N}_{D1})  +  H( X^{N}_{R2,1:\bar{n}}| Y^{N}_{D2})   \label{eq:2ChainRule976}  \\ 
& \leq    H(Y^{N}_{D1})     + H(Y^{N}_{D2})   + H(Y^{N}_{R1}| X^{N}_{S1,1:m})  \nonumber\\ &\quad  +  H( Y^{N}_{R2}|X^{N}_{S2,1:m})  \!+\! H( X^{N}_{R1,1:\bar{n}}|Y^{N}_{D1})  \!+\!  H( X^{N}_{R2,1:\bar{n}}| Y^{N}_{D2})  \nonumber\\ &\quad -  N R_1 -  N R_2 + 2  N \epsilon_N  \label{eq:sumratefinal234}  
\end{align}
where \eqref{eq:2fano302} follows from Fano's inequality; %, and $\epsilon_N \rightarrow 0$ when $N\rightarrow \infty$;  
\eqref{eq:2MarkovYS} follows from the fact that $W_i \rightarrow Y^{N}_{Ri} \rightarrow  Y^{N}_{Di}$ forms a Markov chain ($Y^{N}_{Di}$ is a function of $Y^{N}_{Ri}$), for $i=1,2$;
\eqref{eq:equal236} results from  the  Claim~\ref{claim:m0} below, i.e., $H( Y^{N}_{R1}| W_1)  = H( X^{N}_{S2,1:m}| W_1) $ and $H( Y^{N}_{R2}| W_2)  = H( X^{N}_{S1,1:m}| W_2) $ for this setting without feedback cross-link;
\eqref{eq:addentropy}  follows from the fact that adding information increases the entropy; 
\eqref{eq:2ChainRule} follows from the chain rule for entropy;
\eqref{eq:2entropyadd} uses the fact that adding information increases the mutual information;
\eqref{eq:function326}  follows from the fact that $X^{N}_{Si,1:m}$ is a function of $\{ X^{N}_{Ri,1:\bar{n}}, W_i\}$ for this setting without feedback cross-link, for $i=1,2$;
\eqref{eq:ChainRule3837} follows from the chain rule; 
\eqref{eq:W1W2recon332} results from the fact that the knowledge of $\{W_2, W_1\}$ implies the knowledge of $X^{N}_{R1,1:\bar{n}}$ and $X^{N}_{R2,1:\bar{n}}$, and that  conditioning reduces entropy;
\eqref{eq:2ChainRule564} follows from the chain rule;
\eqref{eq:2W1W2ind} results from the fact that the two messages $W_1 $ and $W_2$ are independent, and the fact that  the knowledge of $\{W_2, W_1\}$ implies the knowledge of $Y^{N}_{D1}$ and $Y^{N}_{D2}$;
\eqref{eq:sumratefinal234} follows from Fano's inequality, i.e., $H(W_i|Y^{N}_{Di}) \leq N\epsilon_N$, and $H(W_i) =N R_i$, for $i=1,2$.

From \eqref{eq:sumratefinal234}, we consequently have 
\begin{align}
& N (2R_1+ 2 R_2-   3\epsilon_N)  \nonumber\\
& \leq   H(Y^{N}_{D1})   \!+\! H(Y^{N}_{D2})   \!+\!  H(Y^{N}_{R1}| X^{N}_{S1,1:m})   \!+\! H( Y^{N}_{R2}|X^{N}_{S2,1:m})  \nonumber\\ &\quad+ H( X^{N}_{R1,1:\bar{n}}|Y^{N}_{D1})  +  H( X^{N}_{R2,1:\bar{n}}| Y^{N}_{D2})   \\
& \leq   Nf  + Nf  +   N\max\{ n-m,  m \}  +  N\max\{ n-m,  m \}    \nonumber\\ &\quad  +  N( \bar{n} -f )^{+} + N ( \bar{n} -f )^{+} \label{eq:sumratefinal2753}  
\end{align}
which yields the bound in \eqref{eq:GeneralOB6} when  $N$ is large.

\begin{claim} \label{claim:m0}
For the setting without feedback cross-link, we have $H( Y^{N}_{R1}| W_1)  = H( X^{N}_{S2,1:m}| W_1) $ and $H( Y^{N}_{R2}| W_2)  = H( X^{N}_{S1,1:m}| W_2) $.
\end{claim}
\begin{proof}
For the setting without feedback cross-link ($\bar{m} =0$), we have
\begin{align}
&H(Y^{N}_{R1}| W_1) \nonumber\\
&=  \sum  H(Y_{R1}[t]| W_1,Y^{t-1}_{R1})   \nonumber\\
&=  \sum  H(Y_{R1}[t]| W_1,Y^{t-1}_{R1}, X^{t}_{R1}, X^{t+1}_{S1})   \label{eq:function983} \\
&=  \sum  H(X_{S2,1:m}[t]| W_1,Y^{t-1}_{R1}, X^{t}_{R1}, X^{t+1}_{S1})   \label{eq:function3432} \\
&=  \sum  H(X_{S2,1:m}[t]| W_1,Y^{t-1}_{R1}, X^{t}_{R1}, X^{t+1}_{S1}, X^{t-1}_{S2,1:m})   \label{eq:function7743} \\
&=  \sum  H(X_{S2,1:m}[t]| W_1, Y^{t-1}_{R1}, X^{t+1}_{S1}, X^{t-1}_{S2,1:m})   \label{eq:function2248}   \\
&=  \sum  H(X_{S2,1:m}[t]| W_1,  X^{t+1}_{S1}, X^{t-1}_{S2,1:m})   \label{eq:function9623}\\
&=  \sum  H(X_{S2,1:m}[t]| W_1, X^{t-1}_{S2,1:m})   \label{eq:function1086}\\
&= H(X^{N}_{S2,1:m}| W_1)     
\end{align}
where \eqref{eq:function983} follows from the fact that  $X_{R1}[t]$  is a function of $Y^{t-1}_{R1}$, and that  $X_{S1}[t+1]$  is a function of $\{W_1, X^{t}_{R1}\}$ for this setting without feedback cross-link;
\eqref{eq:function3432} and \eqref{eq:function7743} follow from the fact that $Y_{R1}[t]$ is an injective function of  both $  X_{S1,1:n}[t]$ and   $X_{S2,1:m}[t]$ given the other;
\eqref{eq:function2248} results from the fact that $X^{t}_{R1}$ is a function of $ Y^{t-1}_{R1}$;
\eqref{eq:function9623} follows from the fact that $Y^{t-1}_{R1}$ is a function of $  \{X^{t-1}_{S1,1:n},  X^{t-1}_{S2,1:m}\}$;
\eqref{eq:function1086} follows from the fact that $X^{t+1}_{S1}$ is a function of $  \{Y_{S1}^{t},  W_1\}$, which in turn is a function of $  \{X_{R1}^{t},  W_1\}$ (since there is no feedback cross-link), and this is in turn a function of $  \{Y_{R1}^{t-1},  W_1\}$ and this is a function of $  \{X^{t-1}_{S2,1:m},  W_1\}$.

Similarly, we have $H( Y^{N}_{R2}| W_2)  = H( X^{N}_{S1,1:m}| W_2) $ due to the symmetry.
\end{proof}

\section{Conclusions}  \label{sec:conclusion}
In this work, we have investigated the capacity benefits that come from overhearing other links for obtaining feedback in wireless networks. We have focused on a deterministic two-hop interference channel where sources in the first hop can overhear the relay transmissions in the next hop which can serve as a feedback signal that can potentially increase the capacity of the first hop. The main challenge in this setup is to devise communication strategies that allow the relays to optimally feedforward information over the second hop while feeding back information over the first hop. We have developed  capacity achieving strategies for various regimes of this channel and showed that while these strategies achieve the performance with dedicated feedback in certain regimes, in certain other regimes, while overheard feedback is still useful, the performance is strictly inferior to the case when there is dedicated feedback links.

\section{Appendix - Achievability proof details  \label{sec:schemedetails}}

In this appendix we provide the details of the three achievability schemes overviewed in Section~\ref{sec:scheme} and prove Theorem~\ref{thm:Achievability} and Theorem~\ref{thm:m0}. Before going into the details of these schemes, we first introduce three classes of information bits for each source.
\bit
\item \emph{Private} bits: these are the information bits of each source that are decodable  by the relay in direct-link but not decodable by the  relay in the cross-link. 
\item \emph{Cooperative common} bits:  these are the information bits of each source that need to be learned by the other source node in order to improve the first hop capacity. Moreover, these bits are decodable  by both relays at the end of the communication process.
\item \emph{Non-cooperative common} bits:  these are the information bits of each source that are decodable  by both relays, but they do not need to be learned by the other source node.
\eit

\subsection{Utilizing cross-link overhearing when $\alpha \leq 2/3$: The scheme $\Xc_{\text{\text{FBXw}}}$}

As briefly described in Section~\ref{sec:schFBXw}, the scheme $\Xc_{\text{\text{FBXw}}}$ achieving rate $R_{\text{sum}}^{\text{FBXw}}$ in \eqref{eq:achieFBXw} is designed for the case when the first hop is in the weak interference regime ($\alpha \leq 2/3$) and makes use of the cross-link feedback for increasing the capacity of the first hop.

Our scheme operates in packets of $R_{\text{sum}}^{\text{FBXw}}$ bits per S-D pair. Each packet is associated with four phases. Among these phases, Phase 1 and Phase 4
involve communication over the first hop, while Phase 2 and Phase 3
involve communication over the second hop (and also the backward
interference channel since it represents the overheard signals over
the second hop). See Fig.~\ref{fig:BlocksFWFB}. 
At the end of these four
phases, the $R_{\text{sum}}^{\text{FBXw}}$ bits of each S-D pair can be decoded by its
corresponding relay. The relaying of these $R_{\text{sum}}^{\text{FBXw}}$ bits to the final
destinations is partially accomplished during the Phases 2 and 3 of
the current packet and the rest is completed during Phases 2 and 3 of
a future packet.
 
We next describe the four phases associated with a packet. In addition to the example shown in Fig.~\ref{fig:Scheme1Example12} and \ref{fig:Scheme1Example12FB}, we also provide here another typical example shown in Fig.~\ref{fig:Example471} and \ref{fig:Example471FB}.

\begin{figure}
\centering
\includegraphics[width=8.9cm]{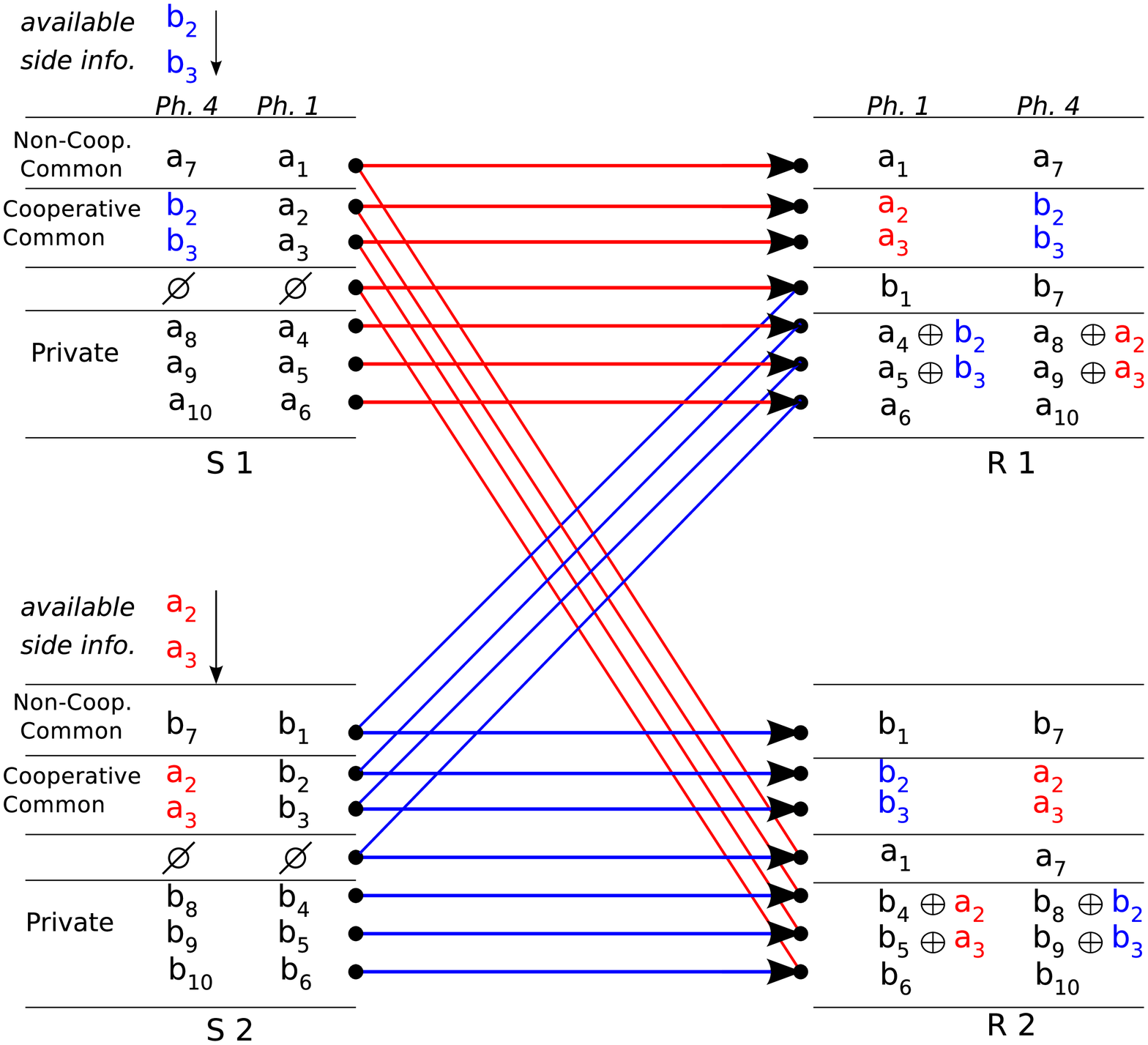}
\caption{The first hop transmission (phase~1 and phase~4) illustration  for scheme $\Xc_{\text{\text{FBXw}}}$ ($m=4, n=7, \bar{m} = \bar{n}=1, f=5$). }
\label{fig:Example471}
\end{figure}

\begin{figure}
\centering
\includegraphics[width=8.9cm]{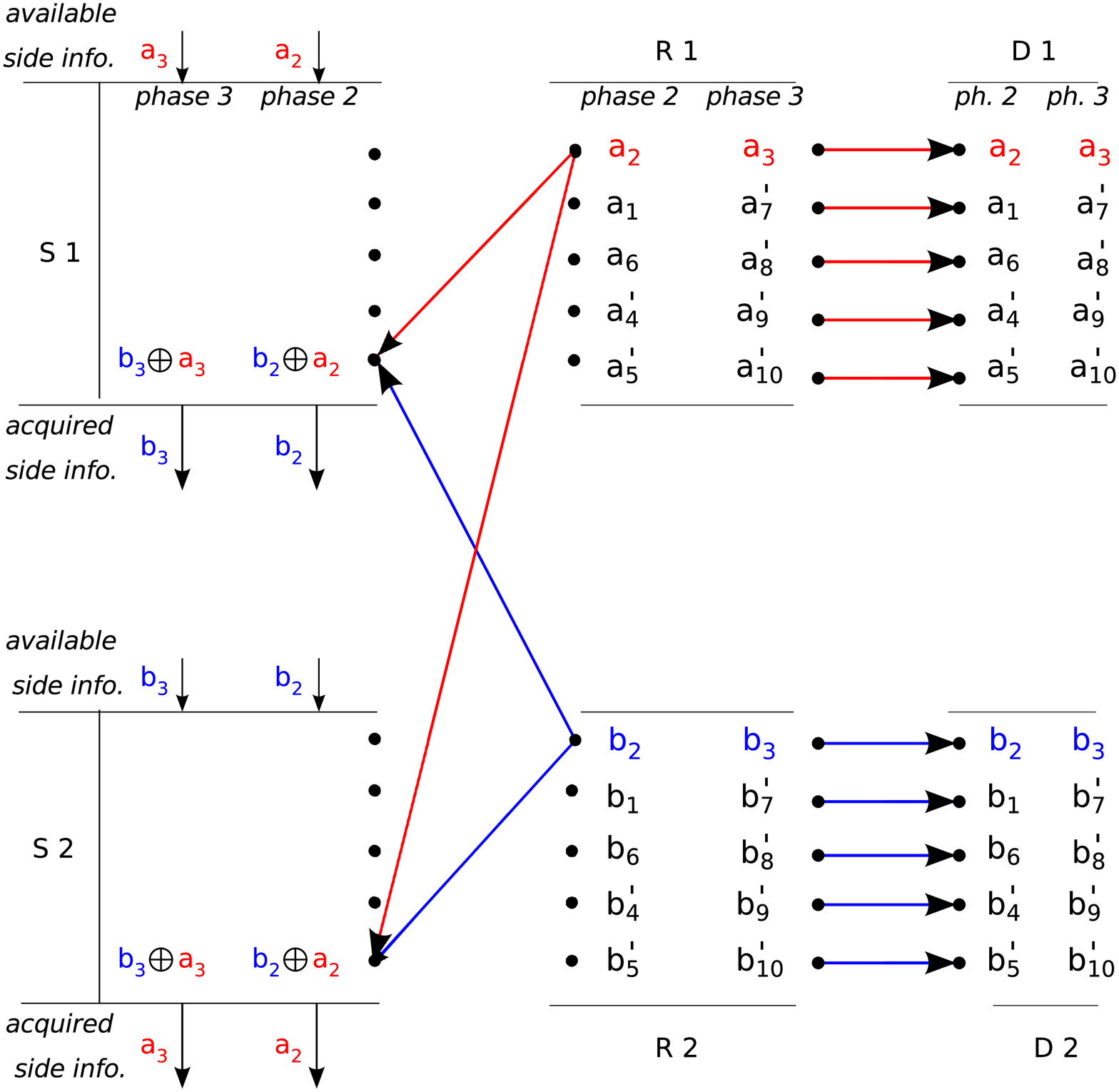}
\caption{The second hop transmission (phase~2 and phase~3) and overhearing illustration  for scheme $\Xc_{\text{\text{FBXw}}}$ ($m=4, n=7, \bar{m} = \bar{n}=1, f=5$). } 
\label{fig:Example471FB}
\end{figure}

\textit{Phase 1:} Each source sends \[\min \bigl\{(2m-n)^{+}, \  f \bigr\}\] \emph{non-cooperative common} bits over the uppermost signal levels,  and \[\min\bigl\{(n-m),  \ \bigl(f -(2m-n)^{+} \bigr)^{+}   \bigr\}\] \emph{private} bits over the lowermost signal
levels, as well as
\begin{align*}
 \min\Bigl\{ 2\bar{m},    \   2n- m -  2\max\{ n-m, m\}, \\  (2f - 2\max\{ n-m, m\}  )^{+}  \Bigr \} 
 \end{align*}
\emph{cooperative common} bits over the signal levels below that of  non-cooperative common bits. For the example shown in Fig.~\ref{fig:Example471},  S1 sends non-cooperative common bit $a_1$, cooperative common bits $a_2, a_3$ and private bits $a_4, a_5, a_6$ over the corresponding signal levels, while S2 sends its information bits in a similar way.
One can see that, the cooperative and  non-cooperative common bits of each source node are received interference free at the relay in the direct-link, while the private bits, that are only visible to the relay in direct-link, arrive interfered  partially with the cooperative common bits of the other source node. 
For the example shown in Fig.~\ref{fig:Example471},  the common bits $a_1, a_2, a_3$ of S1 are received interference free at R1, while the private bits $a_4, a_5$ of S1 are  received at R1 interfered  with cooperative common bits $b_2, b_3$ of S2.
%One can see that, the cooperative and  non-cooperative common bits of S1 (respectively S2) are received interference free at R1 (respectively R2), while the private bits of S1 (respectively S2), which are only visible to R1 (respectively R2), arrive interfered  partially with the cooperative common bits of the other source node. 
%One can see that, the cooperative and  non-cooperative common bits of S1 and of S2 are received interference free at the corresponding relays, while the private bits of S1 and of S2, which are only visible to the corresponding relays, arrive interfered  (partially) with the cooperative common bits of the other source node. 
In the next two phases, each source node will learn the interfering (cooperative common) bits of the other source node, which will allow to resolve the interference in the fourth phase.

\textit{Phases 2 and 3:} At the beginning of these phases the relays
have recovered the cooperative and  non-cooperative common bits of their respective source nodes. 
Since the cooperative common bits need to be both forwarded to their destination and fed
back to the corresponding source node to resolve interference, we use the uppermost level of the relay signals to accomplish both of these goals at the same time. For the example shown in Fig.~\ref{fig:Example471FB}, R1 transmits the cooperative common bits $a_2$ and $a_3$ over its uppermost signal level in Phase 2 and Phase 3 respectively, while R2 transmits similarly $b_2$ and $b_3$. One can see that, each source node can decode the cooperative common bits of the other source node from the overheard signal by using its side information. 
The remaining bit levels of the forward channel are utilized to send non-cooperative common bits and private bits, that have been decoded by the relays but have not been forwarded to the their destinations yet, and that are either from the current packet or from the packet last before.  For the example shown in Fig.~\ref{fig:Example471FB},  the remaining bit levels of R1 are utilized to send $a_1, a_6 $ of the current packet, and $a'_4, a'_5, a'_7,a'_8, a'_9, a'_{10}$ of  the packet last before.

\textit{Phase 4:} 
After decoding the the cooperative common bits that were sent from the other source node at Phase~1 of the current packet,  each source sends fresh $\min\{(2m-n)^{+}, \ f\}$ \emph{non-cooperative common} bits over the uppermost signal levels,  and fresh $\min\bigl\{(n-m),  \ \bigl(f -(2m-n)^{+} \bigr)^{+}   \bigr\}$ \emph{private} bits over the lowermost signal
levels, as well as  the decoded \emph{cooperative common} bits of the other source node over the signal levels below that of  non-cooperative common bits.  See Fig.~\ref{fig:Example471}.
Note that the transmission of the cooperative common bits of the other source node allows to resolve the interference at Phase~1. 
For the example shown in Fig.~\ref{fig:Example471},  the transmission of the cooperative common bits $b_2, b_3$ of S2 from S1 allows to resolve the interference $b_2, b_3$ at R1, while  the transmission of $a_2, a_3$ of S1 from S2 allows to resolve the interference $a_2, a_3$ at R2.
Note that while R1 and R2 can decode all the information bits of their source nodes of the current packet, part of these bits have not been communicated yet to their final destinations. This is accomplished in the Phases 2 and 3 corresponding to the packet after next. Similarly, partial bits of the packet last before were forwarded to their destinations in Phase 2 and 3 corresponding to the current packet. See the timeline in Fig.~\ref{fig:BlocksFWFB}.

With this strategy,  $R_{\text{sum}}^{\text{FBXw}} $  bits of independent information of each source can be communicated to their destination in every 2 channel uses (assuming a large number of packets and ignoring the effect of the first two and the last two packets), i.e., for the case with $\alpha \in [0,  2/3]$ we have
\begin{align*}
& 2\min\Bigl\{(2m-n)^{+}, f\Bigr\}  \\ & +  2\min\Bigl\{(n-m),  \ \bigl(f -(2m-n)^{+} \bigr)^{+}  \Bigr \}    \nonumber\\ & +  \min \Bigl\{ 2\bar{m},       2n - m -  2\max\{ n-m, m \},  \\&  \quad \quad \quad \quad (2f - 2\max\{ n-m, m\}  )^{+} \Bigr\}   \\
=&
\begin{cases}
     2 f    &     \text{if } \      f \leq \max\{ n-m,m\}  \\
 \min \Bigl\{  2\max\{n-m,m \}+2 \bar{m}  , & \\ \quad \quad\quad  2n- m ,  \   2 f \Bigr\}   &     \text{if } \      f \geq  \max\{ n-m,m\}  
\end{cases}
\\  = & \min \Bigl\{  2\max\{n-m,m \}+2 \bar{m}  , \  2n- m ,  \   2 f \Bigr\}       \\
 = &  R_{\text{sum}}^{\text{FBXw}}.
\end{align*} 
%\begin{align}
% R_{\text{sum}}^{\text{FBXw}} = & 2\min\Bigl\{(2m-n)^{+}, f\bigr\}    \nonumber\\& +  2\min\Bigl\{(n-m),  \ \bigl(f -(2m-n)^{+} \bigr)^{+}  \Bigr \}    \nonumber\\ & +  \min\Bigl\{ 2\bar{m},    \   2n- m -  2\max\{ n-m, m \}, \nonumber\\ & \quad \quad \quad \quad   (2f - 2\max\{ n-m, m\}  )^{+}  \Bigr\}  \nonumber\\ = & \min \Bigl\{  2\max\{n-m,m \}+2 \bar{m}  , \  2n- m ,  \   2 f \Bigr\}      \label{eq:FBXwproof}
%\end{align}
%(see Appendix~\ref{sec:FBXwproof} for more details). 
This yields a sum rate of $R_{\text{sum}}^{\text{FBXw}}$ bits per channel use, which turns out to be optimal in the regime of $(\alpha \in [0,  2/3]$, $\bar{m} \geq \bar{n})$. 
See Theorems~\ref{thm:outerbound} and \ref{thm:Achievability}.

\subsection{Utilizing direct-link overhearing when $\alpha \leq 2/3, \bar{m}  = 0$: The scheme $\Xc_{\text{\text{RSw}}}$}

As briefly described in Section~\ref{sec:schRSw}, the scheme $\Xc_{\text{\text{RSw}}}$ achieving rate $R_{\text{sum}}^{\text{RSw}}$ in \eqref{eq:achieRSw} is designed for the case when the first hop is in the weak interference regime and there is no cross-link in the backward
interference channel ($\alpha \leq 2/3, \bar{m} = 0 $). 
In contrast to the previous scheme, this scheme makes use of the direct-link feedback
for increasing the capacity of the first hop. In this case there is a
tradeoff between utilizing the relay bits for conveying information to
the destinations versus feeding back information to the source nodes,
which results in rate splitting. 

As before, this scheme operates in packets, and each packet is
associated with four phases, \emph{i.e.,} Phase 1 and Phase 4 involve
communication over the first hop while Phase 2 and Phase 3 involve
communication over the second hop and the backward channel. See
Fig.~\ref{fig:BlocksFWFB}. 
At the end of these four phases, the $R_{\text{sum}}^{\text{RSw}}/2$ bits of each S-D pair can be
decoded by its corresponding relay. 
The relaying of these $R_{\text{sum}}^{\text{RSw}}/2$ bits to
the final destination will again be accomplished partially in the
Phases 2 and 3 of the current packet, and the rest in Phases 2 and 3
of the packet after next.

We next describe the four phases associated with a packet. In addition to the example shown in Fig.~\ref{fig:Scheme3EgPh1} and \ref{fig:Scheme3EgFB}, we also provide here another typical example ($m=4, n=7, \bar{n}=1, \bar{m} = 0,  f=5$) shown in Fig.~\ref{fig:RSwExample471} and \ref{fig:RSwExample471FB}.

\begin{figure}
\centering
\includegraphics[width=8.9cm]{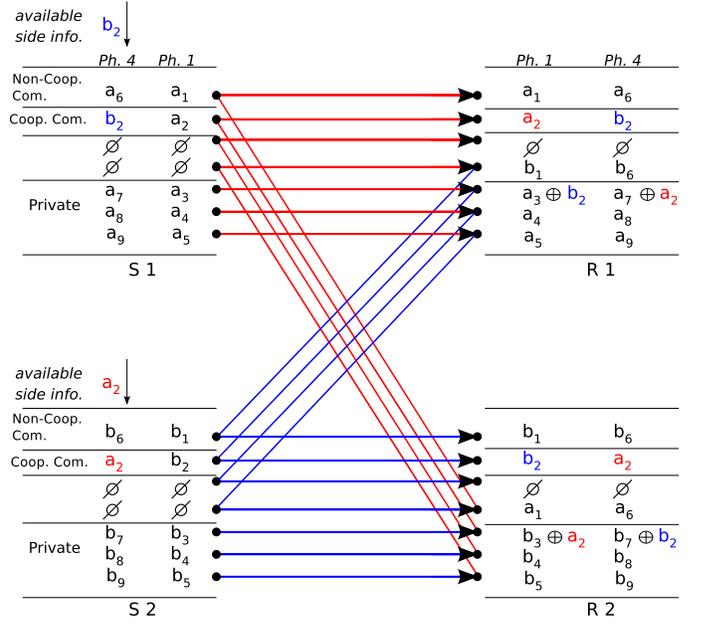}
\caption{The first hop transmission (phase~1 and phase~4) illustration  for scheme $\Xc_{\text{\text{RSw}}}$ ($m=4, n=7, \bar{n}=1, \bar{m} = 0,  f=5$). }
\label{fig:RSwExample471}
\end{figure}

\begin{figure}
\centering
\includegraphics[width=8.9cm]{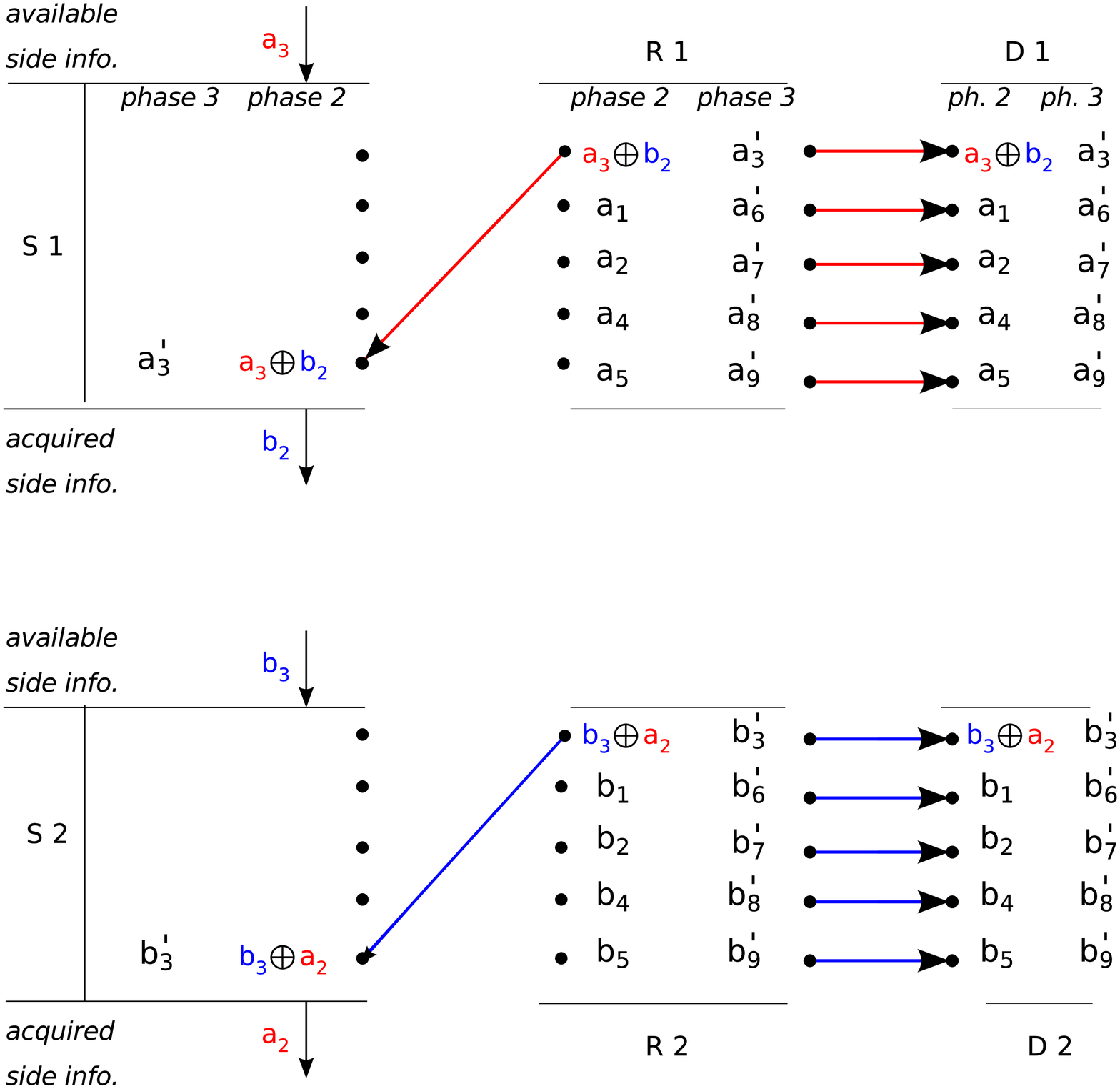}
\caption{The second hop transmission (phase~2 and phase~3) and overhearing illustration  for scheme $\Xc_{\text{\text{RSw}}}$ ($m=4, n=7, \bar{n}=1, \bar{m} = 0,  f=5$). } 
\label{fig:RSwExample471FB}
\end{figure}

\textit{Phase 1:} Each source sends \[\min \bigl\{(2m-n)^{+}, f \bigr\}\] \emph{non-cooperative common} bits over the uppermost signal levels,  and \[\min\Bigl\{(n-m),  \bigl(f -(2m-n)^{+} \bigr)^{+}   \Bigr\}\] \emph{private} bits over the lowermost signal levels, as well as
\begin{align*}
\min\Bigl\{    2( \bar{n} -f )^{+}+ 2   f^{*}  ,    \   2n- m -  2\max\{ n-m, m\}, \\ (2f - 2\max\{ n-m, m\}  )^{+}  \Bigr \} 
\end{align*}
 \emph{cooperative common} bits over the signal levels below that of  non-cooperative common bits, where 
\begin{align} \label{eq:fdefwregion}
f^{*} \defeq \min \Bigl\{  \frac{ \bigl(f- \max\{ n-m,m \}  -  ( \bar{n} -f )^{+} \bigr)^{+} }{ 2}, \nonumber \\  \bar{n} - ( \bar{n} -f )^{+} ,  f   \Bigr\} 
\end{align}
For the example shown in Fig.~\ref{fig:RSwExample471},  S1 sends non-cooperative common bit $a_1$, cooperative common bit $a_2$ and private bits $a_3, a_4, a_5$ over the corresponding signal levels, while S2 sends its information bits in a similar way.
One can see that, the cooperative and  non-cooperative common bits of each source node are received interference free at the relay in direct-link, while the private bits, that are only visible to the relay in direct-link, arrive interfered  partially with the cooperative common bits of the other source node. 
For the example shown in Fig.~\ref{fig:RSwExample471},  the common bits $a_1, a_2$ of S1 are received interference free at R1, while the private bits $a_3$ of S1 are  received at R1 interfered  with cooperative common bits $b_2$ of S2.
%One can see that, the cooperative and  non-cooperative common bits of S1 and of S2 are received interference free at the corresponding relays, while the private bits of S1 and of S2, which are only visible to the corresponding relays, arrive interfered  (partially) with the cooperative common bits of the other source node. See Fig.~\ref{fig:RSwExample471}. 
In the next two phases, each source node will learn the interfering (cooperative common) bits of the other source node, which will allow to resolve the interference in the fourth phase.

\textit{Phase 2 and 3:} At the beginning of  Phase~2, each relay has recovered the cooperative common bits, non-cooperative common bits and partial private bits of the respective source,  and these bits need to be forwarded to the respective destination.
As before, each source needs to learn the common information of the other source node in order to resolve interference. However, since $\bar{m}=0$ this information needs to be  fed back  over the backward direct-link rather than the backward cross-link. 
For the example shown in Fig.~\ref{fig:RSwExample471FB}, R1 and R2 feed back $a_3\oplus b_2$ and  $b_3\oplus a_2$, received in the previous phase,  to S1 and S2 respectively through the backward  direct-link (uppermost signal level). Then the cooperative common bits $b_2$ and $a_2$ can be decoded by S1 and S2 respectively by using the side information at each source node ($a_3$ of S1, $b_3$ of S2). 
The remaining bit levels of the forward channel are utilized to send non-cooperative common bits and private bits, that have been decoded by the relays but have not been forwarded to the their destinations yet, and that are either from the current packet or from the packet last before.  For the example shown in Fig.~\ref{fig:RSwExample471FB},  the remaining bit levels of R1 are utilized to send $a_1, a_2, a_4,a_5$ of the current packet, and $a'_3, a'_6, a'_7,a'_8, a'_9$ of  the packet last before.
Note that in this case the relay bits are split between feedforward and feedback communication. As illustrated  in Fig.~\ref{fig:RSwExample471FB}, in Phase 2 the uppermost bit level of the relay is used solely for feedback and does not provide any useful information to the corresponding destination, while in Phase 3 this bit level is used for sending information to the destination and does not provide any useful feedback information for the source node.

\textit{Phase 4:} Having learned the cooperative common bits of the other source node, each source sends fresh $\min\{(2m-n)^{+}, f\}$ \emph{non-cooperative common} bits over the uppermost signal levels,  and fresh $\min\bigl\{(n-m),  \bigl(f -(2m-n)^{+} \bigr)^{+}   \bigr\}$ \emph{private} bits over the lowermost signal levels, as well as  the decoded \emph{cooperative common} bits of the other source node over the signal levels below that of  non-cooperative common bits. 
See Fig.~\ref{fig:RSwExample471}.
Note that the transmission of the cooperative common bits of the other source node allows to resolve the interference at Phase~1. 
For the example shown in Fig.~\ref{fig:RSwExample471},  the transmission of the cooperative common bits $b_2$ of S2 from S1 allows to resolve the interference $b_2$ at R1, while  the transmission of $a_2$ of S1 from S2 allows to resolve the interference $a_2$ at R2.
Note that while R1 and R2 can decode all the information bits of their source nodes of the current packet, part of these bits have not been communicated yet to their final destinations. This is accomplished in the Phases 2 and 3 corresponding to the packet after next. Similarly, partial bits of the packet last before were forwarded to their destinations in Phase 2 and 3 corresponding to the current packet. %See the timeline in Fig.~\ref{fig:BlocksFWFB}.

\begin{figure}
\centering
\includegraphics[width=8.7cm]{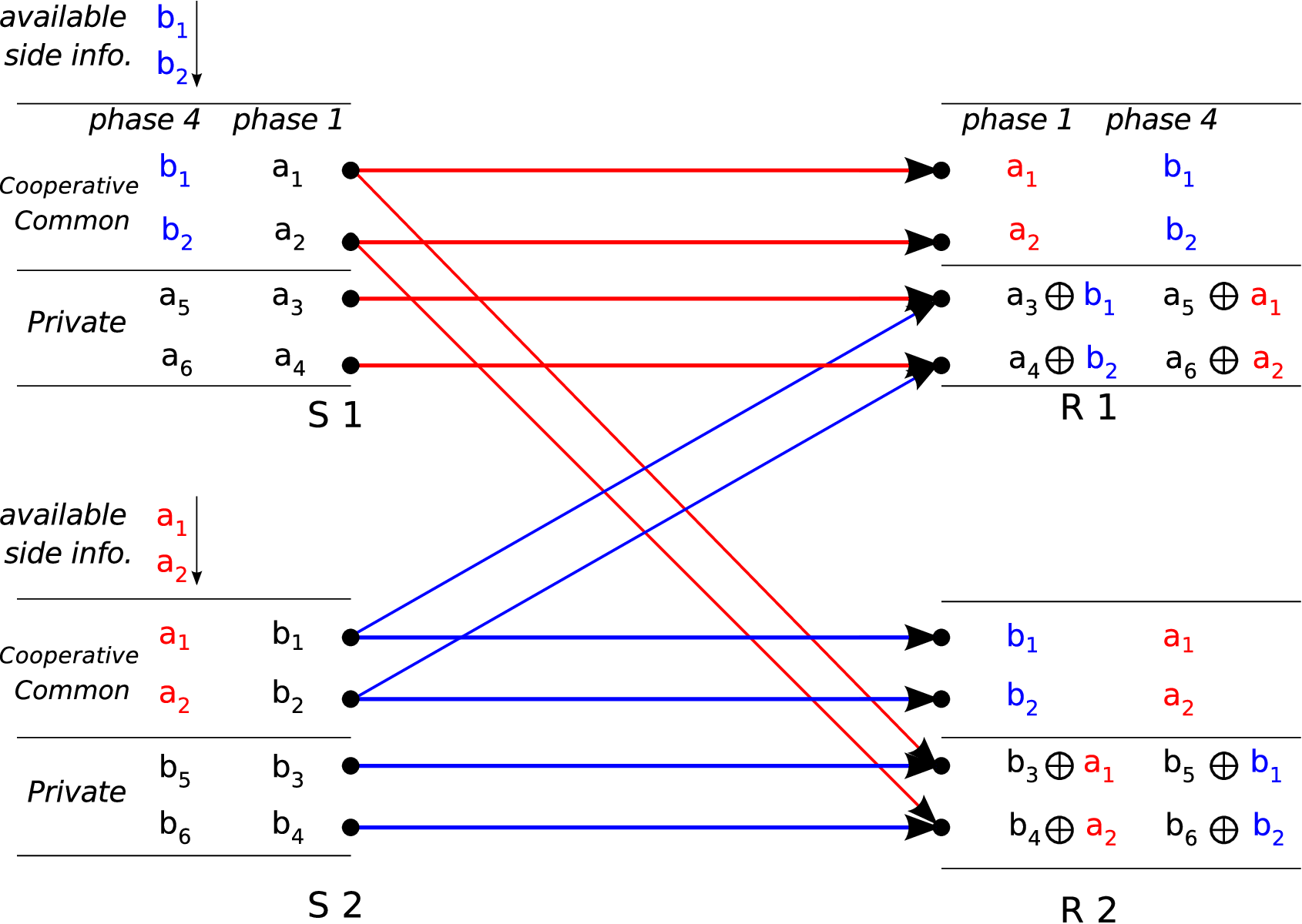}
\caption{The first hop transmission (phase~1 and phase~4) illustration  for scheme $\Xc_{\text{\text{RSw}}}$ ($m=2, n=4, f=3,  \bar{n} =4, \bar{m} = 0 $). }
\label{fig:RSwExample2434}
\end{figure}

\begin{figure}
\centering
\includegraphics[width=8.7cm]{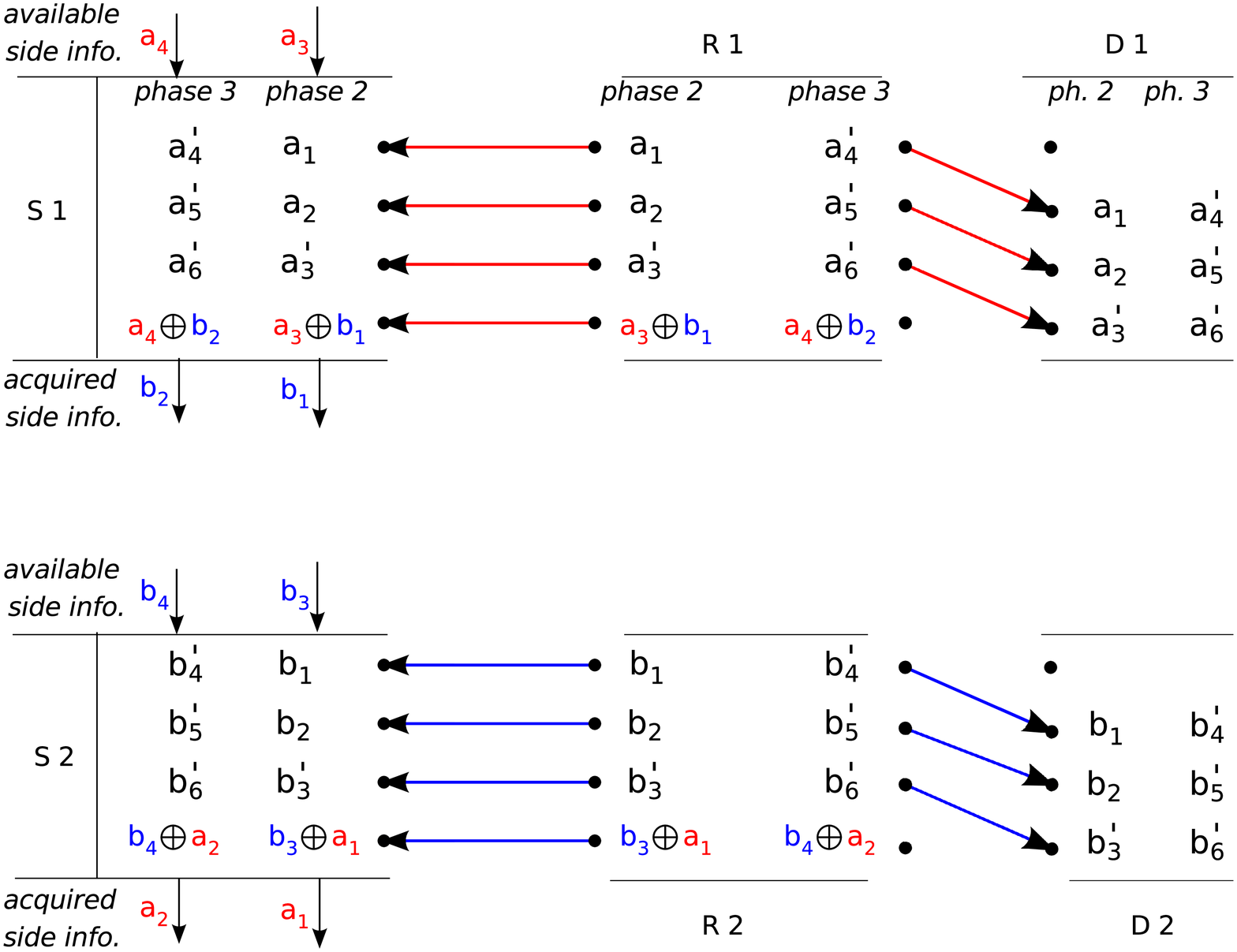}
\caption{The second hop transmission (phase~2 and phase~3) and overhearing illustration  for scheme $\Xc_{\text{\text{RSw}}}$ ($m=2, n=4, f=3,  \bar{n} =4, \bar{m} = 0 $). } 
\label{fig:RSwExample2434FB}
\end{figure}

With this strategy,  the total $R_{\text{sum}}^{\text{RSw}}$ independent bits of each source can be communicated to their destination in every 2 channel uses, 
i.e., for the case with $\alpha \in [0,  2/3]$ we have
\begin{align*}
&  2\min\Bigl\{(2m-n)^{+}, f\Bigr\}  \nonumber\\ & +  2\min\Bigl\{(n-m),  \bigl(f -(2m-n)^{+} \bigr)^{+}  \Bigr \}    \nonumber\\ & +    \min\Bigl\{    2( \bar{n} -f )^{+}+ 2   f^{*}  ,    \   2n- m -  2\max\{ n-m, m\},   \\&  \quad\quad \quad\quad  (2f - 2\max\{ n-m, m\}  )^{+}  \Bigr \} \\
=&
\begin{cases}
     2 f    &   \!\!\!\!\!  \text{if } \      f \leq \max\{ n-m,m\}  \\
     2 f    &   \!\!\!\!\!   \text{if } \     \max\{ n-m,m\}  \leq  f \leq   \Delta_0     \\
   \min \{ f + \Delta_0 , \  2n-m \}   &  \!\!\!\!\!    \text{if } \      f \geq  \Delta_0  \\ & \!\!\!\!\!  \text{and} \   f\!\geq\! \bar{n} \!-\! ( \bar{n} \!-\! f )^{+} \!\geq \!     \frac{ (f\!- \Delta_0 )^{+} }{ 2} \\
    \min \{ 2n- m  , & \\   2 \max\{ n-m,m\}  +  2 \bar{n}  \}   &  \!\!\!\!    \text{if } \      f \geq  \Delta_0 \\&  \!\!\!\! \text{and} \   f\!\geq\!   \frac{ (f\!- \Delta_0)^{+}}{ 2}  \!\geq \!    \bar{n} \!-\! ( \bar{n} \!-\! f )^{+}  \\
     \min \{f + \Delta_0 , \  2n-m \}   &   \!\!\!\!\!   \text{if } \      f \geq  \Delta_0 \\&   \!\!\!\!\!  \text{and} \  \bar{n} \!-\! ( \bar{n} \!-\! f )^{+} \geq f\!\geq\!  \frac{ (f\!- \Delta_0 )^{+} }{ 2}  
\end{cases}
\\  = &\min \Bigl\{  f+ \max\{ n-m,m \}  + (  \bar{n} - f )^{+} ,   \\  &   \quad \quad \quad 2\max\{ n-m,m \} +  2\bar{n},      \ 2n-m ,  \  2f   \Bigr\}     \\
 = & R_{\text{sum}}^{\text{RSw}}
\end{align*}
where $\Delta_0\defeq \max\{ n-m,m\} +  ( \bar{n} -f )^{+} $ and $f^{*}$ is defined in  \eqref{eq:fdefwregion}. 
%i.e., 
%\begin{align}
%& R_{\text{sum}}^{\text{RSw}}  \nonumber\\ =& 2\min\Bigl\{(2m-n)^{+}, f\Bigr\}   \nonumber\\ & +  2\min\Bigl\{(n-m),  \bigl(f -(2m-n)^{+} \bigr)^{+}  \Bigr \}    \nonumber\\ &  +    \min\Bigl\{    2( \bar{n} -f )^{+}+ 2   f^{*}  ,       2n- m -  2\max\{ n-m, m\},  \nonumber\\ & \quad\quad \quad \quad    (2f - 2\max\{ n-m, m\}  )^{+}  \Bigr \}  \nonumber\\
%  = &\min \Bigl\{  f+ \max\{ n-m,m \}  + (  \bar{n} - f )^{+} ,   \nonumber\\ & \quad   \quad  \  2\max\{ n-m,m \} +  2\bar{n},      \ 2n-m ,  \  2f   \Bigr\}     \label{eq:RSwproof}
%\end{align}
%(see Appendix~\ref{sec:proofofRSwproof} for more details).
This yields a sum rate of $R_{\text{sum}}^{\text{RSw}}$ bits per channel use, which turns out to be optimal in the regime of $(\alpha \in [0,  2/3]$, $\bar{m} =0)$. 
See Theorem~\ref{thm:m0}.

\begin{remark} 
In addition to the example shown in Fig.~\ref{fig:Scheme3EgPh1} and \ref{fig:Scheme3EgFB}, and the example shown in Fig.~\ref{fig:RSwExample471} and \ref{fig:RSwExample471FB}, we also provide another typical example (with $\bar{n} > f$) shown in Fig.~\ref{fig:RSwExample2434} and \ref{fig:RSwExample2434FB}.   
In the last example, since $(\bar{n} - f)^{+}$ lowermost levels of the relay signal are visible to the  source node but not to the destination,  those $(\bar{n} - f)^{+}$ lowermost levels are utilized only to feed back the bits containing cooperative common information. As shown in  Fig.~\ref{fig:RSwExample2434FB}, $a_3\oplus b_1$ and  $a_4\oplus b_2$  (the bits containing the cooperative common information $b_1$ and $b_2$ respectively) are fed back from the lowermost levels of the R1 signal, and  $b_3\oplus a_1$ and  $b_4\oplus a_2$ are fed back from the lowermost levels of the R2 signal.  This is different from that in the first two examples with $\bar{n} \leq f$, where the bits containing cooperative common information are fed back from the uppermost levels of the relay signals such that those bits can be overheard by the corresponding sources (cf. Fig.~\ref{fig:Scheme3EgFB} and Fig~\ref{fig:RSwExample471FB}).
\end{remark}

\subsection{Utilizing direct-link overhearing when  $\alpha \geq 2$: The scheme  $\Xc_{\text{\text{RSs}}}$}

As briefly described in Section~\ref{sec:schRSs},  the scheme $\Xc_{\text{\text{RSs}}}$ achieving rate $R_{\text{sum}}^{\text{RSs}}$ in \eqref{eq:achieRSs} is designed for the case when the first hop is in the strong interference regime  ($\alpha \geq 2$). 
Similarly to the scheme $\Xc_{\text{\text{RSw}}}$, this four-phase scheme is based on direct-link overhearing and it uses rate splitting at the relay for feedback and feedforward to achieve the optimal capacity.
At the end of the four phases of a packet,  $\Xc_{\text{\text{RSs}}}/2$ bits of each S-D pair can be
decoded by its corresponding relay. The relaying of these $\Xc_{\text{\text{RSs}}}/2$ bits to
the final destination will again be accomplished partially in the
Phases 2 and 3 of the current packet, and the rest in Phases 2 and 3
of the packet after next.

We next describe the four phases associated with a packet. In addition to the example shown in Fig.~\ref{fig:Scheme2EgPh1} and Fig.~\ref{fig:Scheme2EgFB}, we also provide here another typical example ($m=4, n=1, f=2,  \bar{n} =3, \bar{m} = 1$) shown in Fig.~\ref{fig:RSsExample41231} and \ref{fig:RSsExample41231FB}.

\textit{Phase 1:} Each source sends \[\min \bigl\{ n, \ f \bigr\}\] \emph{non-cooperative common} bits over its uppermost signal levels, as well as \[\min\Bigl\{    2( \bar{n} -f )^{+} +  2   f^{'}  ,  m- 2n, \ (2f - 2n )^{+}  \Bigr \} \]   \emph{cooperative common} bit over the signal levels below that of the non-cooperative common bits, where 
\begin{align} \label{eq:fRSs}
f^{'} \defeq  \min \Bigl\{  \frac{ \bigl(f- n  -  ( \bar{n} -f )^{+} \bigr)^{+} }{ 2},   \bar{n} - ( \bar{n} -f )^{+} ,  f   \Bigr\} 
\end{align}
The rest bit levels are not utilized and are fixed to $0$. 
For the example shown in Fig.~\ref{fig:RSsExample41231},  S1 sends non-cooperative common bit $a_1$ and cooperative common bits $a_2, a_3$ over the upper three signal levels respectively, while S2 sends $b_1$ and $b_2,b_3$ in a similar way.
One can see that, the non-cooperative common bit, $a_1$ of S1 and $b_1$ of S2, is received interference free at both relays,  while the cooperative common bits, $a_2,a_3$ of S1 and $b_2,b_3$ of S2,  are received (interference free) at the relay in cross-link but not at the relay in direct-link.
 In order to improve the first hop capacity,  in the next two phases the  cooperative common bits $a_2,a_3$ will be fed back from R2 to S2 such that it can be sent from S2 to R1 in the fourth phase,  while the cooperative common bit  $b_2,b_3$ will be fed back from R1 to S1 such that it can be sent from S1 to R2 in the fourth phase.

\textit{Phase 2 and 3:} At the beginning of  Phase~2, each relay has recovered all  non-cooperative common bits of its respective source  that need to be forwarded to their destination, and has recovered all cooperative common bits of the other source that need to be  fed back through the backward  direct-link.   
For the setting with $\bar{n} \leq f$ (cf. Fig.~\ref{fig:Scheme2EgPh1} and~\ref{fig:Scheme2EgFB}), the cooperative common bits are fed back from the uppermost levels of the relay signal, such that those bits can be overheard by the source node. For the other setting with $\bar{n} > f$ (cf.  Fig.~\ref{fig:RSsExample41231} and \ref{fig:RSsExample41231FB}), the cooperative common bits are fed back from the lowermost levels of the relay signal, such that the lowermost $(\bar{n}  - f)^{+}$ levels of the relay signal --- that are are visible to the respective source node but not to the destination ---  can be utilized properly (i.e., used for feeding back information only).
For the example shown in Fig.~\ref{fig:RSsExample41231FB}, R1 transmits the cooperative common bits $b_2$ and $b_3$ over its lowermost signal
level in Phase 2 and Phase 3 respectively, while R2 transmits
similarly $a_2$ and $a_3$.   
The remaining signal levels of the relay signal are utilized to forward the decoded non-cooperative common bits,  $a_1, a'_2,  a'_3, a'_4$ of S1 and $b_1, b'_2,  b'_3, b'_4$ of S2, to its destination, where  the bits $a'_2,  a'_3, a'_4, b'_2,  b'_3, b'_4$  belong to the packet last before and have not been forwarded yet.  
Note that in this case the relay bits are split between feedforward and feedback communication. 

 \begin{figure}
\centering
\includegraphics[width=8.5cm]{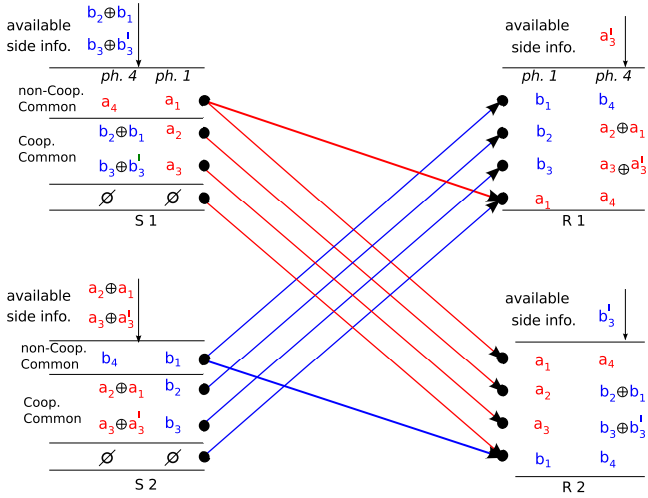}
\caption{The first hop transmission (phase~1 and phase~4) illustration  for scheme $\Xc_{\text{\text{RSs}}}$ ($m=4, n=1, f=2,  \bar{n} =3, \bar{m} = 1 $). }
\label{fig:RSsExample41231}
\end{figure}

\begin{figure}
\centering
\includegraphics[width=8.7cm]{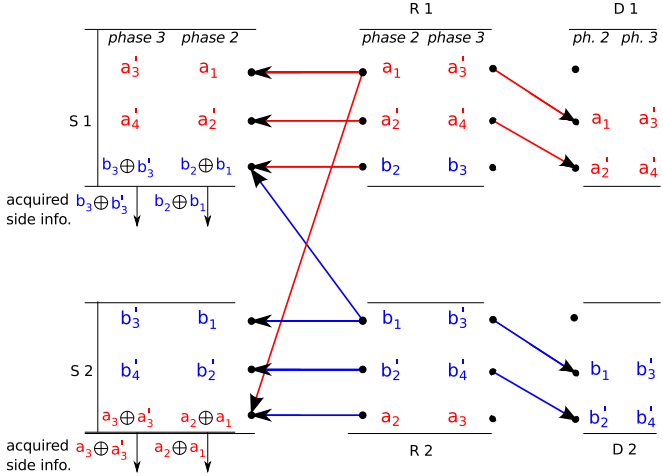}
\caption{The second hop transmission (phase~2 and phase~3) and overhearing illustration  for scheme $\Xc_{\text{\text{RSs}}}$ ($m=4, n=1, f=2,  \bar{n} =3, \bar{m} = 1 $). } 
\label{fig:RSsExample41231FB}
\end{figure}

\textit{Phase 4:}  In this phase each source sends \emph{fresh} $\min \bigl\{ n, \ f \bigr\}  $  non-cooperative common bits to the relay in direct-link and sends \emph{non-fresh} $\min\Bigl\{    2( \bar{n} -f )^{+} +  2   f^{'}  ,  m- 2n, \ (2f - 2n )^{+}  \Bigr \}$  bits containing previous cooperative common information of the other source to the relay in cross-link.      
For the example shown in Fig.~\ref{fig:RSsExample41231} and \ref{fig:RSsExample41231FB},  after Phase~3 S1 has received the linear combinations $b_2\oplus b_1$ and  $b_3\oplus b'_3$ that need to be sent to R2 such that the cooperative common bits $b_2$ and $b_3$ can be recovered by R2 with the prior knowledge of $ b_1$ and  $b'_3$.  Similarly,  S2 has received the linear combinations $a_2\oplus a_1$ and  $a_3\oplus a'_3$ that need to be sent to R1.
Then during Phase~4 S1 transmits the fresh  bit $a_4$ and  the combinations $b_2\oplus b_1$ and  $b_3\oplus b'_3$  from its three uppermost signal levels  respectively,  while S2 sends the fresh  bit $b_4$ and  the combinations $a_2\oplus a_1$ and  $a_3\oplus a'_3$ in a similar way. 
Note that this allows R1 to learn  $a_2, a_3$ from the loop $S1 \rightarrow R2 \rightarrow S2 \rightarrow R1$, and allows R2 to learn $b_2, b_3$ from the loop $S2 \rightarrow R1 \rightarrow S1 \rightarrow R2$. Note that while R1 and R2 can decode $\{a_2, a_3, a_4\}$ and $\{b_2,b_3,b_4\}$ respectively, these bits have not been communicated yet to their final destinations. This is accomplished in the Phase~2 and Phase 3 corresponding to the packet after next.
Note that, the encoding and decoding of the general scheme follow similarly from that of the examples shown in  Fig.~\ref{fig:RSsExample41231}, \ref{fig:RSsExample41231FB} and in  Fig.~\ref{fig:Scheme2EgPh1}, \ref{fig:Scheme2EgFB}.

With this strategy,  the total $R_{\text{sum}}^{\text{RSs}}$ independent bits of each source can be communicated to their destination in every 2 channel uses, i.e., for the case with $\alpha\geq 2 $ we have 
\begin{align}
&  2\min \bigl\{ n, \ f \bigr\}  \nonumber\\ & + \min\Bigl\{    2( \bar{n} -f )^{+} +  2   f^{'}  ,  m- 2n, \ (2f - 2n )^{+}  \Bigr \}  \nonumber\\
=&
\begin{cases}
     2 f    &     \text{if } \   Con1      \\
     2 f    &     \text{if } \    Con2   \\
   \min \{ n+ f +  ( \bar{n} -f )^{+} , \  m \}   &     \text{if } \  Con3     \\
    \min \{  2n + 2 \bar{n} , \  m \}   &     \text{if } \   Con4    \\
     \min \{ n+ f +  ( \bar{n} -f )^{+} , \  m \}   &     \text{if } \ Con5   
\end{cases} \nonumber\\ 
 = &\min \{  \ n + f + ( \bar{n} -f )^{+},     \ 2n+ 2\bar{n},    \  m , \  2f  \}         \nonumber\\
 = & R_{\text{sum}}^{\text{RSs}}    \nonumber
\end{align}
where 
\begin{align*}
Con1 &:= \{f \leq n  \}  \\ 
Con2 &:=\{ n \leq  f \leq n+  ( \bar{n} -f )^{+} \} \\ 
Con3 &:=\Bigl\{ f \geq n+  ( \bar{n} -f )^{+}  \\ &\quad   \text{and} \   f\geq \bar{n} - ( \bar{n} - f )^{+} \geq      \frac{ (f-n -  ( \bar{n} - f )^{+} )^{+} }{ 2} \Bigr \}\\ 
Con4 &:=\Bigl\{ f \geq n+  ( \bar{n} -f )^{+}   \\ &\quad   \text{and} \   f\geq   \frac{ (f-n -  ( \bar{n} - f )^{+} )^{+} }{ 2}  \geq     \bar{n} - ( \bar{n} - f )^{+} \Bigr \}\\ 
Con5 &:=\Bigl\{ f \geq n+  ( \bar{n} -f )^{+} \\ &\quad   \text{and} \  \bar{n} - ( \bar{n} - f )^{+} \geq f\geq   \frac{ (f-n -  ( \bar{n} - f )^{+} )^{+} }{ 2}    \Bigr \}
\end{align*}
and $f^{'}$ is defined in \eqref{eq:fRSs}.
%\begin{align}
% R_{\text{sum}}^{\text{RSs}} =& 2\min \bigl\{ n, \ f \bigr\}  \nonumber\\& + \min\Bigl\{    2( \bar{n} -f )^{+} +  2   f^{'}  ,  m- 2n, \ (2f - 2n )^{+}  \Bigr \} \nonumber\\
%  = & \min \{  \ n + f + ( \bar{n} -f )^{+},     \ 2n+ 2\bar{n},    \  m , \  2f  \}      \label{eq:RSsproof}
%\end{align}
%(see Appendix~\ref{sec:proofofRSsproof} for more details).
This yields a sum rate of $R_{\text{sum}}^{\text{RSs}}$ bits per channel use, which turns out to be optimal in the regime of $ \alpha  \geq 2$. 
See Theorems~\ref{thm:outerbound} and \ref{thm:Achievability}.

%\bibliographystyle{IEEEtran}
%\bibliography{IEEEabrv,final_refs}


\end{document}